\documentclass[preprint,authoryear,12pt]{elsarticle}

\makeatletter
 \def\ps@pprintTitle{%
 \let\@oddhead\@empty
 \let\@evenhead\@empty
 \def\@oddfoot{\centerline{\thepage}}%
 \let\@evenfoot\@oddfoot}
 \makeatother

 \usepackage{graphics}
\usepackage{subfig}
\usepackage{tikz}
\usetikzlibrary{arrows,decorations.pathmorphing,backgrounds,positioning,fit,petri}

\usepackage{amssymb}
\usepackage{amsmath}
\usepackage{amsthm}
\usepackage{dsfont}
\newtheorem{thm}{Theorem}
\newcommand{\argmax}{\operatornamewithlimits{argmax}}

\journal{Computational Statistics \& Data Analysis}

\begin{document}

\begin{frontmatter}



\title{A unified framework for model-based clustering, linear regression and multiple cluster structure detection}


 \author{Giuliano Galimberti}
 \author{Annamaria Manisi}
 \author{Gabriele Soffritti\corref{cor1}}
  \ead{gabriele.soffritti@unibo.it}

\cortext[cor1]{Correspondence to: Department of Statistical Sciences, University of Bologna \\ via Belle Arti 41, 40126 Bologna, Italy. Tel.: +39 051 2098193, Fax: +39 051 232153}


\address{Department of Statistical Sciences, University of Bologna, Italy}

\begin{abstract}

A general framework for dealing with both linear regression and clustering problems is described. It includes Gaussian clusterwise linear regression analysis with random covariates and cluster analysis via Gaussian mixture models with variable selection. It also admits a novel approach for detecting multiple clusterings from possibly correlated sub-vectors of variables, based on a model defined as the product of conditionally independent Gaussian mixture models. A necessary condition for the identifiability of such a model is provided. The usefulness and effectiveness of the described methodology are illustrated using simulated and real datasets.
\end{abstract}


\end{frontmatter}



\section{Introduction}\label{sec:intro}

Cluster analysis and regression are two fields of statistics useful for unsupervised and supervised statistical learning, respectively \citep{hastie2009}. Although their goals are quite different, there may be situations in which relevant information contained in a given dataset can be obtained using both types of analysis.

Consider, for example, a dataset characterised by the presence of an unknown cluster structure (a grouping of the given set of observations into clusters), but suppose that this structure is confined in a sub-space of the variable space.
Since the effect of the presence of uninformative variables is a masking of the cluster structure \citep[see, e.g.,][]{gordon1999}, the use of methods able to select the informative variables in a cluster analysis, such as those proposed by \citet{fowlkes1988}, \citet{gnanadesikan1995}, \citet{brusco2001}, \citet{montanari2001}, \citet{fraiman2008}, \citet{steinley2008a} and \citet{witten2010}, is crucial for a proper recovery of the unknown cluster structure from the observed data. If the variables are continuous, Gaussian mixture models can be employed. In these models it is assumed that the joint probability density function (p.d.f.) of the variables is a mixture of $K$ ($K \geq 1$) Gaussian components (one component for each cluster) \citep[see, e.g.,][]{mclachlan2000,melnykov2010}. The number of components is generally chosen through model selection criteria, such as the Bayesian information criterion ($BIC$) \citep{schwarz1978}. Methods that simultaneously select the informative variables and find the cluster structure, based on Gaussian mixture models, have been proposed by \citet{dy2004}, \citet{law2004}, \citet{tadesse2005}, \citet{raftery2006},
\citet{pan2007}, \citet{xie2008}, \citet{wang2008},
\citet{galimberti2009}, \citet{maugis2009a}, \citet{maugis2009b}, \citet{zeng2009},
\citet{zhou2009}, \citet{guo2010} and \citet{andrews2014}.
Comparisons among some of these methods, based on analyses of simulated and real datasets, can be found in \citet{steinley2008b}, \citet{witten2010}, \citet{celeux2011}, \citet{celeux2014} and \citet{andrews2014}. In particular, the variable selection methods proposed by \citet{raftery2006} are based on a model in which the vector of the examined variables is assumed to be partitioned into two sub-vectors: one is composed of the informative variables and the other contains the uninformative ones. Furthermore, it is assumed that the joint p.d.f.~is the product of a Gaussian mixture model (with $K \geq 2$ components) for the distribution of the informative variables and a Gaussian linear regression model for the conditional distribution of the uninformative variables given the informative ones. The selection of the informative variables is then recast as a model comparison problem, based on the $BIC$. A local optimum in model space is found through a greedy search algorithm. In the model selected using these methods, the uninformative variables are assumed to linearly depend on all the informative ones. In some situations this assumption may be restrictive and could lead to erroneous model selections and, consequently, to erroneous results. To avoid this drawback, two more versatile models have been proposed: one allows the uninformative variables to be explained by only a subset of the informative ones \citep{maugis2009a}; the other also allows
that some uninformative variables are independent of all the informative ones \citep{maugis2009b}. The three methods just described make use of Gaussian linear regression models for performing variable selection in model-based cluster analysis. Thus, they embed a supervised learning process into an unsupervised one.

Consider another example in which the dataset provides information about responses and predictors for a regression problem. Furthermore, suppose that the sample observations come from a population composed of several sub-populations, each of which is characterised by a different regression equation, but the information about which sub-population each observation comes from is unknown. An approach able to deal with such an unobserved heterogeneity in a regression problem is represented by clusterwise regression. In this approach the dependence of the responses on the predictors is described by means of a finite mixture \citep[see, e.g.,][]{quandt1978,desarbo1988,deveaux1989}. With continuous responses Gaussian mixture models are generally employed, resulting in mixtures of Gaussian linear regression models. In an analysis based on these models, an unsupervised learning process is embedded into a supervised one.

In addition to the variable selection, another problem to be tackled when performing a cluster analysis is that datasets may also be characterised by the presence of several unknown cluster structures, that is different groupings of the same set of observations defined in different subspaces of the variable space. Since a classical assumption
in most clustering methods is that one single cluster structure is contained in the data, some relevant information
about the ways observations are clustered could be missed. Methods to deal with this problem are due, for example, to \citet{soffritti2003}, \citet{friedman2004}, \citet{belitskaya2006}, \citet{galimberti2007}, \citet{poon2013}, \citet{dangbailey2015} and \citet{liu2015}. In the latter four papers, model-based methods are described. In particular, the solution proposed by \citet{galimberti2007} relies on a model in which the vector of the observed variables is assumed to be partitioned into independent sub-vectors, and a Gaussian mixture model is specified for each variable sub-vector.

In this paper a general framework is developed by exploiting the idea of combining the above mentioned unsupervised and supervised methods. More specifically, it incorporates Gaussian model-based clustering with variable selection, Gaussian clusterwise linear regression (with a constraint on the regression coefficients) and a novel approach for detecting multiple cluster structures from possibly correlated variable sub-vectors. These three fields of statistics are brought together through the specification of a general and flexible model. The basic idea of the proposed framework is introduced in Section~\ref{sec:intro_example}, while the theory is illustrated in Section~\ref{sec:framework}. Namely, the general model is presented in Section~\ref{sec:model}. A theorem providing necessary conditions for its identifiability is proved in Section~\ref{sec:identifiability}. Details about the maximum likelihood (ML) estimation of the model parameters and the issue of model selection are given in Sections~\ref{sec:estimation} and \ref{sec:model_selection}, respectively.
Section~\ref{sec:results} contains experimental results obtained by exploring the model space in some simulated and real datasets through genetic algorithms. Concluding notes are reported in Section~\ref{sec:conclusions}.

\section{An introductory example}\label{sec:intro_example}

The basic idea of the framework developed in this paper is illustrated through an example referring to a dataset containing information about 17 male and 16 female students in an introductory statistics class at a US college. The dataset is available at \verb"http://economics-files.pomona.edu/GarySmith/" \verb"StatSite/eclectic.html". The example focuses on three variables: student's height, mother's height and father's height. Figure~\ref{fig:scatterplot_height} shows the scatterplot matrix, where male and female students are represented using squares and circles, respectively.

\begin{figure}
 \centering
   \includegraphics[width=0.8\textwidth]{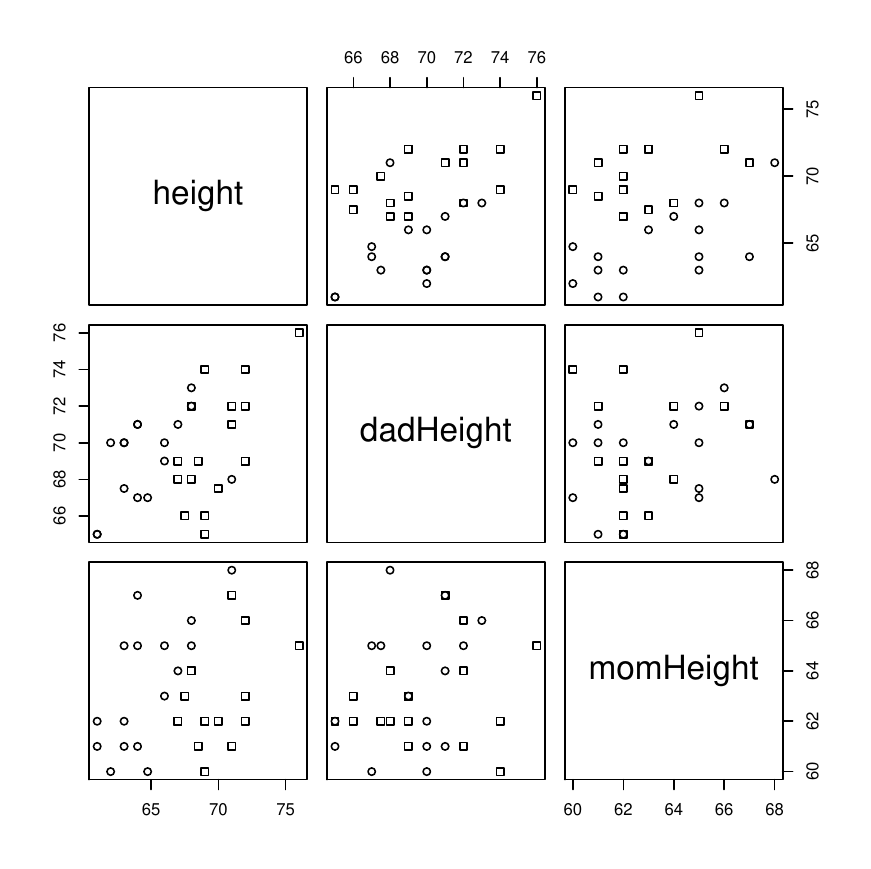}
  \caption{Scatterplot matrix of the heights (in inches) in the students' dataset.}\label{fig:scatterplot_height}
\end{figure}

From an unsupervised point of view, one may be interested in detecting whether the dataset is characterised by the presence of a cluster structure in the joint distribution of the three examined variables. Using the package \verb"mclust" \citep{fraley2002, fraley2012} for the \verb"R" software environment \citep{R2015}, Gaussian mixture models with $K=1,2,3,4$ components are fitted to the dataset, and the best model is selected according to the $BIC$. Since such a model results to be a Gaussian model with only one component, no evidence of any cluster structure emerges from this unsupervised analysis. This model is denoted as $M_1$ in Section~\ref{sec:ripresa esempio intro}.

As described in Section~\ref{sec:intro}, the presence of a cluster structure may be masked by uninformative variables. Thus, a further analysis is carried out through the \verb"R" package \verb"clustvarsel" that implements the variable selection methods proposed by \citet{raftery2006}. This analysis suggests that parents' heights are informative variables while the students' height can be discarded. Furthermore, the best Gaussian mixture model for parents' heights identifies $K=2$ clusters. Figure~\ref{fig:ellipse_heights} shows the two 85\% ellipses obtained from the mean vectors and covariance matrices of this model. White and black colours are used to highlight observations assigned to different clusters. According to this analysis students with short parents can be separated from students with tall parents.

\begin{figure}
 \centering
   \includegraphics[width=0.7\textwidth]{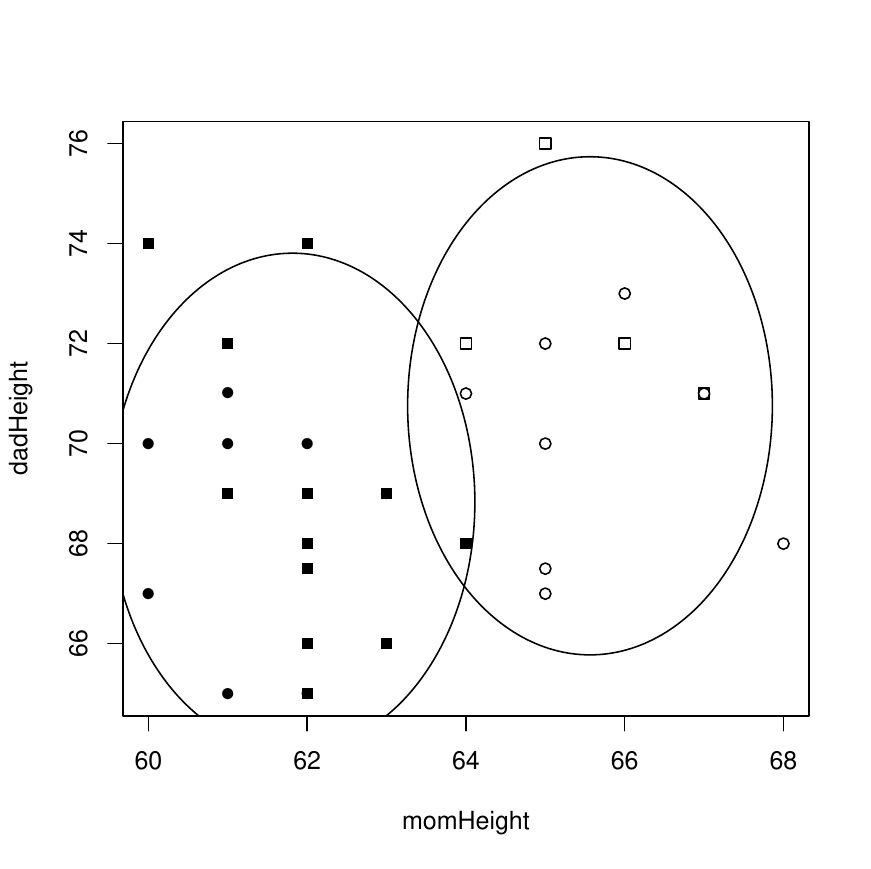}
  \caption{Scatterplot of the mother's and father's heights in the students' dataset with the segmentation of the observations obtained from the best fitted Gaussian mixture model for the parents' heights distribution. The bivariate normal density contours described by the two ellipses have a 85\% probability content.}\label{fig:ellipse_heights}
\end{figure}

In order to measure the association between the detected cluster structure and the gender of the students, the joint classification of the students based on their gender and the segmentation obtained from the parents' heights is examined, and the adjusted Rand index (aRi) \citep{hubert1985} is computed. From the obtained results (see Table~\ref{tab:heights}, left part) it emerges that such an association is very low. The joint model for the three examined variables resulting from this analysis is denoted as $M_2$ in Section~\ref{sec:ripresa esempio intro}.

\begin{table}
\caption{Comparison between the classification of the students based on their gender and the segmentations obtained from models $M_2$ (Structure 1) and $M_5$ (Structure 2).}  \label{tab:heights}
\centering
\begin{tabular}{lcccccc}
  \hline
  & \multicolumn{3}{l}{Structure 1} & \multicolumn{3}{l}{Structure 2}\\   \hline
 &  \multicolumn{2}{l}{Cluster} & & \multicolumn{2}{l}{Cluster} &\\
Gender  & 1  & 2  & &  1   & 2 & \\ \hline
      F & 8 & 8 & &  1 & 15 & \\
      M & 13 & 4 & &  16 & 1 &\\
  \hline
      aRi & \multicolumn{2}{c}{0.047} & & \multicolumn{2}{c}{0.765} &\\
  \hline
\end{tabular}
\end{table}

Given the drawbacks of the variable selection methods proposed by \citet{raftery2006}, further analyses are carried out by using two \texttt{C++}
softwares (\texttt{SelvarClust} and \texttt{SelvarClustIndep}) that incorporate algorithms for fitting and selecting models described in \citet{maugis2009a, maugis2009b}. Parents' heights result to be informative variables also using these techniques (with the same Gaussian mixture model with two components selected by \verb"clustvarsel"). As far as the Gaussian linear regression model for the conditional distribution of the students' heights is concerned, only the father's height is selected as a predictor. The resulting model for the joint distribution of the three heights is $M_3$ in Section~\ref{sec:ripresa esempio intro}.

The examined dataset can be analysed also from a supervised point of view. Namely, the interest could be in predicting the height of a student from his/her parents' heights. In particular, this supervised analysis is performed using two different models for the conditional distribution of the students' heights: a Gaussian linear regression model and a mixture of two Gaussian linear regression models with the same regression coefficients. Both models are intentionally specified and estimated by ignoring the information about gender. The fictitiously unobserved heterogeneity obtained in this way could be captured by using a mixture of two Gaussian linear regression models. For both types of models the fathers' height results to be the only relevant predictor for the students' height. Figure~\ref{fig:regression_heights} shows the regression lines estimated using the two models obtained after performing regression variable selection. Models $M_4$ and $M_5$ in Section~\ref{sec:ripresa esempio intro} denote the models for the joint distribution of the three variables obtained by multiplying each of these regression models by the best Gaussian model for the marginal distribution of the parents' heights. In particular, using a mixture of two Gaussian regressions allows to produce a segmentation of the observations into two clusters. In the right panel of Figure~\ref{fig:regression_heights} white and black colours are used to distinguish between such clusters. Differently from the cluster structure obtained from the unsupervised analysis of the parents' heights, the cluster structure detected in this supervised analysis of the conditional distribution of the students' heights is highly associated with their gender (see Table~\ref{tab:heights}, right part).

\begin{figure}
 \centering
   \includegraphics[width=0.8\textwidth]{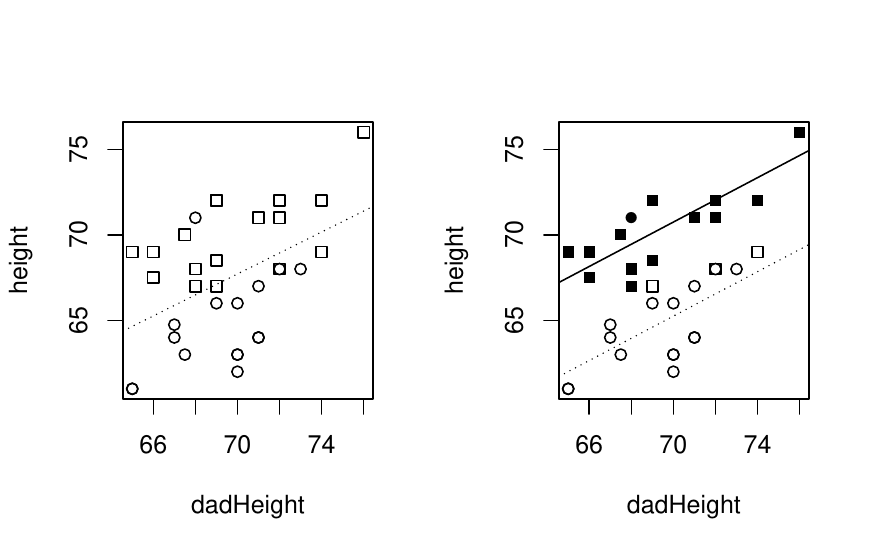}
  \caption{Scatterplot of the student's and father's heights and fitted regression lines using a Gaussian linear model (left panel) and a clusterwise Gaussian linear model with two components having the same regression coefficients (right panel).}\label{fig:regression_heights}
\end{figure}

The examined dataset is characterised by the presence of two different cluster structures that have been detected by combining methods for unsupervised and supervised statistical learning. In particular, the identification of such structures has required the specification of a model for the joint p.d.f.~of the three examined variables obtained as the product of a mixture of two bivariate Gaussian distributions (for the marginal p.d.f.~of the parents' heights) and a mixture of two univariate Gaussian linear regression models with the same regression coefficients (for the conditional p.d.f.~of the students' height). This model is denoted as $M_6$ in Section~\ref{sec:ripresa esempio intro}, where further results obtained from the analysis of this dataset are presented. The basic idea of analysing the same dataset using different types of models in order to extract relevant information from different perspectives is presented in a formal and more general way in Section~\ref{sec:framework}.

\section{A general framework for clustering and regression}\label{sec:framework}

\subsection{Models for clustering and regression}\label{sec:model}
For the ease of exposition assume initially that the examined dataset may be characterised by up to two unknown cluster structures, $S_1$ and $S_2$. Let $\mathbf{X}=(X_1, \ldots, X_L)$ be the random vector of $L$ observed continuous variables and $\mathbf{X}^{S_1}$, $\mathbf{X}^{S_2}$, $\mathbf{X}^{U}$ be a splitting of $\mathbf{X}$ into three non-overlapping sub-vectors, where $\mathbf{X}^{S_2}$ and $\mathbf{X}^{U}$ can be empty. In particular, the sub-vector $\mathbf{X}^{S_g}$ is composed of $L_g$ variables that provide information about the cluster structure $S_g$ ($g=1,2$), while the sub-vector $\mathbf{X}^{U}$ is composed of $L_U$ uninformative variables, with $L_1+L_2+L_U = L$.

The following set of assumptions is specified.

\noindent \textit{A1)} The marginal distribution of $\mathbf{X}^{S_1}$ is given by a Gaussian mixture model. Namely:
\begin{equation}\label{mn1}
f\left(\mathbf{x}^{S_1}; \boldsymbol{\theta}_1\right)=\sum_{k=1}^{K_1} \pi^{(1)}_k \phi_{L_1}\left( \mathbf{x}^{S_1};\boldsymbol{\mu}^{(1)}_k,\boldsymbol{\Sigma}^{(1)}_k\right),
\end{equation}
where $K_1 \geq 1$, $\boldsymbol{\theta}_1=(\boldsymbol{\pi}_1, \boldsymbol{\mu}_1, \boldsymbol{\sigma}_1)$,
$\boldsymbol{\pi}_1=\left(\pi^{(1)}_1,\ldots, \pi^{(1)}_{K_1}\right)$, $\boldsymbol{\mu}_1=\left(\boldsymbol{\mu}^{(1)}_1, \ldots, \boldsymbol{\mu}^{(1)}_{K_1}\right)$,
$\boldsymbol{\sigma}_1=\left(\boldsymbol{\Sigma}^{(1)}_1, \ldots, \boldsymbol{\Sigma}^{(1)}_{K_1}\right)$, and $\phi_{L_1}\left(\cdot;\boldsymbol{\mu}^{(1)}_k,\boldsymbol{\Sigma}^{(1)}_k\right)$ denotes the p.d.f.~of the $L_1$-dimensional normal distribution with mean vector $\boldsymbol{\mu}^{(1)}_k$ and covariance matrix $\boldsymbol{\Sigma}^{(1)}_k$.

\noindent \textit{A2)} The conditional p.d.f.~of $\mathbf{X}^{S_2}$ given $\mathbf{X}^{S_1}$ is equal to
\begin{equation}\label{mn2a}
f\left(\mathbf{x}^{S_2}|\mathbf{x}^{S_1}; \boldsymbol{\theta}_2\right)=\sum_{k=1}^{K_2} \pi^{(2)}_k \phi_{L_2}\left(\mathbf{x}^{S_2};\boldsymbol{\mu}^{(2)}_k,\boldsymbol{\Sigma}^{(2)}_k\right),
\end{equation}
where $K_2 \geq 1$,
\begin{equation}\label{mn2b}
\boldsymbol{\mu}^{(2)}_k=\boldsymbol{\beta}_{0}+\boldsymbol{\lambda}^{(2)}_{k} + \mathbf{B}_{21}\mathbf{x}^{S_1}, \ k=1, \ldots, K_2,
\end{equation}
with $\boldsymbol{\beta}_{0}$ and $\boldsymbol{\lambda}^{(2)}_{k}$ denoting $L_2$-dimensional vectors, and
$\mathbf{B}_{21}$ representing a $L_2 \times L_1$ matrix.
The condition defined by equation (\ref{mn2b})
is equivalent to assuming that the dependence of $\mathbf{X}^{S_2}$ on $\mathbf{X}^{S_1}$
is defined by a multivariate linear regression model whose error terms follow a mixture of $K_2$
Gaussian components. Specifically,
\begin{equation} \label{mixregression_nonid}
\mathbf{X}^{S_2} = \boldsymbol{\beta}_{0} + \mathbf{B}_{21}\mathbf{x}^{S_1} + \boldsymbol{\epsilon}_2, \ \boldsymbol{\epsilon}_2 \sim \sum_{k=1}^{K_2} \pi^{(2)}_k N_{L_2}\left(\boldsymbol{\lambda}^{(2)}_{k},\boldsymbol{\Sigma}^{(2)}_{k}\right),
\end{equation}
where $N_{L_2}\left(\boldsymbol{\lambda}^{(2)}_{k},\boldsymbol{\Sigma}^{(2)}_{k}\right)$ denotes the $L_2$-dimensional normal distribution with mean vector $\boldsymbol{\lambda}^{(2)}_{k}$ and covariance matrix $\boldsymbol{\Sigma}^{(2)}_{k}$. Thus, the conditional distribution of $\mathbf{X}^{S_2}$ given $\mathbf{X}^{S_1}$ is described by a mixture of $K_2$ Gaussian linear regression models with the constraint that the effect of $\mathbf{X}^{S_1}$ on the expected value of $\mathbf{X}^{S_2}$ is the same for all the $K_2$ components of the mixture (\ref{mn2a}). In order to guarantee identifiability of model (\ref{mixregression_nonid}), it is necessary to require some constraints on $\boldsymbol{\beta}_{0}$ or the $\boldsymbol{\lambda}^{(2)}_{k}$'s. Namely, $\boldsymbol{\beta}_{0}=\textbf{0}$ or
$\sum_{k} \pi^{(2)}_k \boldsymbol{\lambda}^{(2)}_{k} = \mathbf{0}$. This problem does not arise if model (\ref{mn2b}) is directly parameterized as follows:
\begin{equation}\label{mn2c}
\boldsymbol{\mu}^{(2)}_k=\boldsymbol{\gamma}^{(2)}_{k}+ \mathbf{B}_{21}\mathbf{x}^{S_1}, \ k=1, \ldots, K_2,
\end{equation}
where $\boldsymbol{\gamma}^{(2)}_{k}=\boldsymbol{\beta}_{0} + \boldsymbol{\lambda}^{(2)}_{k}$. Then, $\boldsymbol{\theta}_2=(\boldsymbol{\pi}_2, \boldsymbol{\gamma}_2, \mathbf{B}_{21}, \boldsymbol{\sigma}_2)$,
$\boldsymbol{\pi}_2=\left(\pi^{(2)}_1,\ldots, \pi^{(2)}_{K_2}\right)$, $\boldsymbol{\gamma}_2=\left(\boldsymbol{\gamma}^{(2)}_1, \ldots, \boldsymbol{\gamma}^{(2)}_{K_2}\right)$,
 $\boldsymbol{\sigma}_2=\left(\boldsymbol{\Sigma}^{(2)}_1, \ldots, \boldsymbol{\Sigma}^{(2)}_{K_2}\right)$ \citep[for further details see][]{soffritti2011}.

\noindent \textit{A3)} The conditional distribution of $\mathbf{X}^{U}$ given $\left(\mathbf{X}^{S_1}, \mathbf{X}^{S_2}\right)$ follows a Gaussian linear regression model. Specifically:
\begin{equation} \label{regression}
\mathbf{X}^{U} = \boldsymbol{\alpha}_{0} + \mathbf{A}_{1}\mathbf{x}^{S_1} + \mathbf{A}_{2}\mathbf{x}^{S_2} + \boldsymbol{\epsilon}_U, \ \boldsymbol{\epsilon}_U \sim N_{L_U}\left(\boldsymbol{0},\boldsymbol{\Sigma}_{U}\right),
\end{equation}
whose parameters are $\boldsymbol{\theta}_U=\left(\boldsymbol{\alpha}_0, \mathbf{A}_1, \mathbf{A}_2,
\boldsymbol{\Sigma}_U\right)$, where $\boldsymbol{\alpha}_{0}$ is a $L_U$-dimensional vector, $\mathbf{A}_1$ and $\mathbf{A}_2$ represent matrices of dimension $L_U \times L_1$ and $L_U \times L_2$, respectively.

Thus, the joint p.d.f.~of $\mathbf{X}$ can be obtained as follows:
\begin{eqnarray} \nonumber
f(\mathbf{x};\boldsymbol{\theta}) & = & \sum_{k=1}^{K_1} \pi^{(1)}_k \phi_{L_1}\left( \mathbf{x}^{S_1};\boldsymbol{\mu}^{(1)}_k,\boldsymbol{\Sigma}^{(1)}_k\right) \\ \nonumber
& \times &
\sum_{k=1}^{K_2} \pi^{(2)}_k \phi_{L_2}\left(\mathbf{x}^{S_2};\boldsymbol{\gamma}^{(2)}_{k}+ \mathbf{B}_{21}\mathbf{x}^{S_1},\boldsymbol{\Sigma}^{(2)}_k\right)  \\ \label{modello nuovo}
& \times & \phi_{L_U}\left(\mathbf{x}^{U}; \boldsymbol{\alpha}_{0} + \mathbf{A}_{1}\mathbf{x}^{S_1} + \mathbf{A}_{2}\mathbf{x}^{S_2},\boldsymbol{\Sigma}_U\right),
\end{eqnarray}
where $\boldsymbol{\theta}= \left(\boldsymbol{\theta}_1, \boldsymbol{\theta}_2, \boldsymbol{\theta}_U\right)$.

Some approaches for unsupervised and/or supervised analysis illustrated in Sections~\ref{sec:intro} and \ref{sec:intro_example} can be obtained from model (\ref{modello nuovo}).
\begin{itemize}
\item
When $\mathbf{X}^{S_2}=\mathbf{X}^{U}=\emptyset$, the Gaussian mixture model used to detect the presence of a cluster structure in the joint distribution of $\mathbf{X}$ is obtained. The graphical representation of such a model is depicted in Figure~\ref{fig:modello1}(a), where $Z$ denotes a nominal latent variable affecting the probability distribution of $\mathbf{X}$. If the number of categories of such a variable is greater than 1, then the joint distribution of $\mathbf{X}$ provides information about a latent cluster structure. Thus, $Z$ can be considered as an indicator of cluster membership. This is the approach to cluster analysis based on Gaussian mixture models.

\item
When $\mathbf{X}^{S_2}=\emptyset$, $K_1 \geq 2$, $\mathbf{X}^{U} \neq \emptyset$, a cluster structure is assumed to be hidden in the variable sub-space $\mathbf{X}^{S_1}$. Specifically, $\mathbf{X}^{S_1}$ and $\mathbf{X}^{U}$ represent the vectors of the informative and uninformative variables, respectively. The model associated with this situation can be graphically represented as in Figure~\ref{fig:modello1}(b). It represents the tool for model-based cluster analysis with variable selection according to the approach proposed by \citet{raftery2006}.

\item
When $K_1=1$, $\mathbf{X}^{S_2} \neq \emptyset$, $\mathbf{X}^{U} = \emptyset$, model~(\ref{modello nuovo}) reduces to either a Gaussian linear regression model (if $K_2=1$) or a mixture of Gaussian linear regression models (if $K_2 >1$) with the same regression coefficients \citep{soffritti2011} and Gaussian random covariates.

\item
When $K_1>1$, $\mathbf{X}^{S_2} \neq \emptyset$, $K_2 >1$, a novel model with two latent cluster structures is obtained. The marginal distribution of the variable sub-vector $\mathbf{X}^{S_1}$ provides information about the first structure while the second structure is hidden in the conditional distribution of $\mathbf{X}^{S_2}$ given $\mathbf{X}^{S_1}$. Thus, a novel approach for detecting two cluster structures from two possibly dependent sub-vectors of random variables is obtained. By setting $\mathbf{B}_{21}=\boldsymbol{0}$ in equation~(\ref{mn2c}), the second latent structure is, in fact, defined in the marginal distribution of $\mathbf{X}^{S_2}$, thus leading to the model proposed by \citet{galimberti2007}. In both cases $\mathbf{X}^{S_1}$ and $\mathbf{X}^{S_2}$ represent vectors of informative variables. If, in addition, $\mathbf{X}^{U} \neq \emptyset$, then $\mathbf{X}^{U}$ is the vector of uninformative variables. In this latter case model~(\ref{modello nuovo}) also allows to perform variable selection. The graphical representation of model (\ref{modello nuovo}) when $(\mathbf{X}^{S_1}, \mathbf{X}^{S_2}, \mathbf{X}^{U})$ is a partition of $\mathbf{X}$ and both $K_1$ and $K_2$ are greater than 1 is given in Figure~\ref{fig:modello1}(c), where $Z_1$ denotes the nominal latent variable (with $K_1$ categories) that affects the probability distribution of $\mathbf{X}^{S_1}$, and $Z_2$ denotes the nominal latent variable (with $K_2$ categories) affecting the probability distribution of $\mathbf{X}^{S_2}|\mathbf{X}^{S_1}$. Thus, in model (\ref{modello nuovo}), the sub-vector $\mathbf{X}^{S_2}$ is assumed to be conditionally independent
of $Z_1$ given $\mathbf{X}^{S_1}$. Furthermore, $Z_1$ and $Z_2$ are assumed to be independent.
\end{itemize}

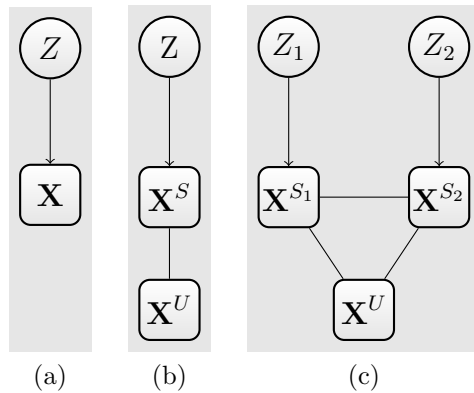
\begin{figure}[h]
\centering
\subfloat[ ][ ]		
{
\begin{tikzpicture}
[latent/.style={circle,draw,thick,inner sep=0pt,minimum size=8mm,top color=white,bottom color=black!10},
manifest/.style={rounded corners,draw,thick,inner sep=0pt,minimum size=8mm,top color=white,bottom color=black!10}]

\node (Z) at ( 0,2) [latent] {$Z$};
\node (X) at ( 0,0.05) [manifest] {$\textbf{X}$};
\node (phantom) at (0,-1.75) [] {}; 

\draw [->] (Z) -- (X);

\begin{scope}[on background layer]
\node [fill=black!10,fit=(Z) (X) (phantom)] {};
\end{scope}

\end{tikzpicture}
}
\, 
\subfloat[ ][ ]		
{
\begin{tikzpicture}
[latent/.style={circle,draw,thick,inner sep=0pt,minimum size=8mm,top color=white,bottom color=black!10},
manifest/.style={rounded corners,draw,thick,inner sep=0pt,minimum size=8mm,top color=white,bottom color=black!10}]

\node (Z) at (0,2) [latent] {Z};
\node (XS) at (0,0) [manifest] {$\textbf{X}^S$};
\node (XU) at (0,-1.5)  [manifest] {$\textbf{X}^U$};

\draw [->] (Z) -- (XS);
\draw [-] (XS) -- (XU);

\begin{scope}[on background layer]
\node [fill=black!10,fit=(Z) (XS) (XU)] {};
\end{scope}

\end{tikzpicture}
}\,
\subfloat[ ][ ]		
{
\begin{tikzpicture}
[latent/.style={circle,draw,thick,inner sep=0pt,minimum size=8mm,top color=white,bottom color=black!10},
manifest/.style={rounded corners,draw,thick,inner sep=0pt,minimum size=8mm,top color=white,bottom color=black!10}]

\node (Z1) at ( -2,2) [latent] {$Z_1$};
\node (Z2) at ( 0,2) [latent] {$Z_2$};
\node (X1) at ( -2,0) [manifest] {$\textbf{X}^{S_1}$};
\node (X2) at ( 0,0) [manifest] {$\textbf{X}^{S_2}$};
\node (Xu) at ( -1,-1.5) [manifest] {$\textbf{X}^U$};

\draw [->] (Z1) -- (X1);
\draw [->] (Z2) -- (X2);
\draw [-] (X1) -- (X2);
\draw [-] (X1) -- (Xu);
\draw [-] (X2) -- (Xu);

\begin{scope}[on background layer]
\node [fill=black!10,fit=(Z1) (Z2) (X1) (X2) (Xu)] {};
\end{scope}

\end{tikzpicture}
}
\caption{Graphical representations of three models obtained from equation (\ref{modello nuovo}).}
\label{fig:modello1}
\end{figure}

The model defined in equation (\ref{modello nuovo}) places two main restrictions. First, it forces the uninformative variables to linearly depend on all the remaining variables (see equation (\ref{regression}). Furthermore, it prevents the possible presence of uninformative variables that are independent of any other variable. Models that do not impose these restrictions can be obtained as follows. Whenever $\mathbf{X}^{U} \neq \emptyset$, let $\mathbf{X}^{S_1}_U \subseteq \mathbf{X}^{S_1}$ and, if $\mathbf{X}^{S_2} \neq \emptyset$, $\mathbf{X}^{S_2}_U \subseteq \mathbf{X}^{S_2}$ denote the sub-vectors of $\mathbf{X}^{S_1}$ and $\mathbf{X}^{S_2}$, respectively, useful to predict the uninformative variables $\mathbf{X}^{U}$ according to a Gaussian linear regression model. Thus, equation (\ref{regression}) is modified as follows:
\begin{equation} \label{regression_gen}
\mathbf{X}^{U} = \boldsymbol{\alpha}_{0} + \mathbf{A}_{1U}\mathbf{x}^{S_1}_U + \mathbf{A}_{2U}\mathbf{x}^{S_2}_U + \boldsymbol{\epsilon}_U, \ \boldsymbol{\epsilon}_U \sim N_{L_U}\left(\boldsymbol{0},\boldsymbol{\Sigma}_{U}\right),
\end{equation}
where $\mathbf{A}_{1U}$ and $\mathbf{A}_{2U}$ are matrices of dimensions $L_U \times L_{1U}$ and $L_U \times L_{2U}$, respectively, with $L_{1U}$ and $L_{2U}$ denoting the lengths of sub-vectors $\mathbf{X}^{S_1}_U$ and $\mathbf{X}^{S_2}_U$.
Furthermore, let $(\mathbf{X}^{S_1}, \mathbf{X}^{S_2}, \mathbf{X}^{U}, \mathbf{X}^{I})$ be a splitting of $\mathbf{X}$ into four non-overlapping sub-vectors, where the additional sub-vector $\mathbf{X}^{I}$ can be empty.
Specifically, $\mathbf{X}^{I}$ contains $L_I$ uninformative variables that are also independent of all the remaining variables. Thus, the resulting p.d.f.~of $\mathbf{X}$ is defined as follows:
\begin{eqnarray} \nonumber
f(\mathbf{x};\boldsymbol{\theta}) & = & \sum_{k=1}^{K_1} \pi^{(1)}_k \phi_{L_1}\left( \mathbf{x}^{S_1};\boldsymbol{\mu}^{(1)}_k,\boldsymbol{\Sigma}^{(1)}_k\right)\\ \nonumber
& \times &
\sum_{k=1}^{K_2} \pi^{(2)}_k \phi_{L_2}\left(\mathbf{x}^{S_2};\boldsymbol{\gamma}^{(2)}_{k}+ \mathbf{B}_{21}\mathbf{x}^{S_1},\boldsymbol{\Sigma}^{(2)}_k\right) \\ \label{modello nuovo_gen}
& \times & \phi_{L_U}\left(\mathbf{x}^{U}; \boldsymbol{\alpha}_{0} + \mathbf{A}_{1U}\mathbf{x}^{S_1}_U + \mathbf{A}_{2U}\mathbf{x}^{S_2}_U,\boldsymbol{\Sigma}_U\right) \times \phi_{L_I}\left(\mathbf{x}^{I}; \boldsymbol{\mu}_{I},\boldsymbol{\Sigma}_I\right),
\end{eqnarray}
where  $\boldsymbol{\theta}= \left(\boldsymbol{\theta}_1, \boldsymbol{\theta}_2, \boldsymbol{\theta}_U, \boldsymbol{\theta}_I\right)$, with $\boldsymbol{\theta}_I=(\boldsymbol{\mu}_{I},\boldsymbol{\Sigma}_I)$.

The graphical representation of the models obtained from equation (\ref{modello nuovo_gen}) by removing the second restriction is given in Figure~\ref{fig:modello2}(a); Figure~\ref{fig:modello2}(b) represents models resulting from the removal of both restrictions. When $\mathbf{X}^{S_2}=\emptyset$, $\mathbf{X}^{U} \neq \emptyset$ and $\mathbf{X}^{I} = \emptyset$, equation (\ref{modello nuovo_gen}) gives the model proposed by \citet{maugis2009a}; if $\mathbf{X}^{I} \neq \emptyset$, the model introduced in \citet{maugis2009b} is obtained.



\begin{figure}[h]
\centering
\subfloat[ ][ ]		
{
\begin{tikzpicture}
[latent/.style={circle,draw,thick,inner sep=0pt,minimum size=8mm,top color=white,bottom color=black!10},
manifest/.style={rounded corners,draw,thick,inner sep=2pt,minimum size=8mm,top color=white,bottom color=black!10},
manifest2/.style={rounded corners,draw,thick}]	

\node (Z1) at ( -2,2) [latent] {$Z_1$};
\node (Z2) at ( 0,2) [latent] {$Z_2$};
\node (X1) at ( -2,0) [manifest] {$\textbf{X}^{S_1}$};
\node (X2) at ( 0,0) [manifest] {$\textbf{X}^{S_2}$};
\node (Xu) at ( -1,-1.5) [manifest] {$\textbf{X}^U$};
\node (Xi) at (1, -1.5) [manifest] {$\textbf{X}^I$};

\node (Xui) [manifest2] [fit=(Xu) (Xi), draw]{};

\draw [->] (Z1) -- (X1);
\draw [->] (Z2) -- (X2);
\draw [-] (X1) -- (X2);
\draw [-] (X1) -- (Xu);
\draw [-] (X2) -- (Xu);

\begin{scope}[on background layer]
\node [fill=black!10,fit=(Z1) (Z2) (X1) (X2) (Xu) (Xi) (Xui)] {};
\node [top color=white, bottom color=black!10,fit=(Xu) (Xi)] {};
\end{scope}

\end{tikzpicture}
}
\,
\subfloat[ ][ ]		
{
\begin{tikzpicture}
[latent/.style={circle,draw,fill,thick,inner sep=0pt,minimum size=8mm,top color=white,bottom color=black!10},
manifest/.style={rounded corners,draw,thick,inner sep=2pt,minimum size=8mm,top color=white,bottom color=black!10},
manifest2/.style={rounded corners,draw,thick}]

\node (Z1) at ( -2.5,2) [latent] {$Z_1$};
\node (Z2) at ( 1.5,2) [latent] {$Z_2$};
\node (X1notu) at ( -3.5,0) [manifest] {$\bar{\textbf{X}}^{S_1}_{U}$};
\node (X1u) at ( -1.5,0) [manifest] {$\textbf{X}^{S_1}_U$};
\node (X2u) at ( 0.5,0) [manifest] {$\textbf{X}^{S_2}_U$};
\node (X2notu) at ( 2.5,0) [manifest] {$\bar{\textbf{X}}^{S_2}_{U}$};
\node (Xu) at ( -0.5,-1.5) [manifest] {$\textbf{X}^U$};
\node (Xi) at ( 1.5,-1.5) [manifest] {$\textbf{X}^I$};

\node (X1) [manifest2] [fit=(X1notu) (X1u), draw]{};
\node (X2) [manifest2] [fit=(X2notu) (X2u), draw]{};
\node (Xui) [manifest2] [fit=(Xu) (Xi), draw]{};

\draw [->] (Z1) -- (X1);
\draw [->] (Z2) -- (X2);
\draw [dashed] (X1u) -- (X1notu);
\draw [dashed] (X2u) -- (X2notu);
\draw [-] (X1u) -- (Xu);
\draw [-] (X2u) -- (Xu);

\draw [-] (X1) -- (X2);

\begin{scope}[on background layer]
\node [fill=black!10,fit=(Z1) (Z2) (X1) (X2) (Xu) (Xi) (Xui)] {};
\node [top color=white, bottom color=black!10,fit=(X1u) (X1notu)] {};
\node [top color=white, bottom color=black!10,fit=(X2u) (X2notu)] {};
\node [top color=white, bottom color=black!10,fit=(Xu) (Xi)] {};
\end{scope}

\end{tikzpicture}
}

\caption{Graphical representations of two models obtained from equation (\ref{modello nuovo_gen}), where $\mathbf{\bar{X}}_U^{S_1}=\mathbf{X}^{S_1} \smallsetminus \mathbf{X}^{S_1}_U$,
  $\mathbf{\bar{X}}_U^{S_2}=\mathbf{X}^{S_2} \smallsetminus \mathbf{X}_U^{S_2}$.}

\label{fig:modello2}
\end{figure}
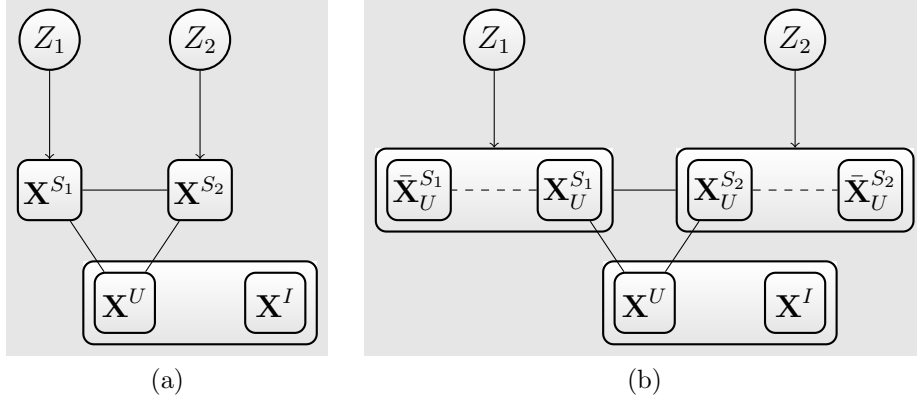

Whenever $\mathbf{X}^{S_2} \neq \emptyset$, further more general models can be defined. For example, in equation (\ref{mixregression_nonid}) the variables in the sub-vector $\mathbf{X}^{S_2}$ are forced to linearly depend on all the variables belonging to $\mathbf{X}^{S_1}$. This restriction can be removed as follows.
Let $\mathbf{X}^{S_1}_{S_2} \subseteq \mathbf{X}^{S_1}$ denote the sub-vector of $\mathbf{X}^{S_1}$ composed of the
predictors of $\mathbf{X}^{S_2}$. Then, $\boldsymbol{\mu}^{(2)}_k$ in equation (\ref{mn2c}) becomes
\begin{equation}\label{mn2c_rid}
\boldsymbol{\mu}^{(2)}_k=\boldsymbol{\gamma}^{(2)}_{k}+ \mathbf{B}^*_{21}\mathbf{x}^{S_1}_{S_2}, \ k=1, \ldots, K_2,
\end{equation}
where the number of columns of matrix $\mathbf{B}^*_{21}$ coincides with the length of $\mathbf{X}^{S_1}_{S_2}$.
If, in addition, $L_2 \geq 2$, it is possible to define a model in which a different sub-vector of the variables belonging to $\mathbf{X}^{S_1}$ can be employed for predicting each variable in $\mathbf{X}^{S_2}$. Such a model can be specified as follows. Let $\mathbf{X}^{S_2}[l]$ be the $l$-th variable of $\mathbf{X}^{S_2}$ and $\mathbf{X}^{S_1}_{l}$ be the sub-vector of $\mathbf{X}^{S_1}$ containing the predictors of $\mathbf{X}^{S_2}[l]$. Then, $\boldsymbol{\mu}^{(2)}_k$ can be obtained as follows:
\begin{equation}\label{mn2c_rid_seemingly}
\boldsymbol{\mu}^{(2)}_k[l]=\boldsymbol{\gamma}^{(2)}_{k}[l]+ \mathbf{x}'^{S_1}_{l} \boldsymbol{\beta}_l, \ l=1, \ldots, L_2,
\end{equation}
where the notation $\mathbf{a}[l]$ is used to denote the $l$-th element of vector $\mathbf{a}$. The joint model for the vector $\mathbf{X}^{S_2}$ resulting from equation~(\ref{mn2c_rid_seemingly}) represents a seemingly unrelated linear regression model with a Gaussian mixture for the errors \citep[for further details see][]{galimberti2015}.

Models (\ref{modello nuovo}) and (\ref{modello nuovo_gen}) can be easily modified so as to admit that the examined dataset may be characterised by up to, say, $G$ unknown cluster structures, where $G$ can be greater than 2. Let $(\mathbf{X}^{S_1}, \mathbf{X}^{S_2}, \ldots, \mathbf{X}^{S_G}, \mathbf{X}^{U})$ be a splitting of $\mathbf{X}$ into $G+1$ non-overlapping sub-vectors, where
$\mathbf{X}^{S_g}, g=2 \ldots, G,$ can be empty, and $L_g$ denotes the length of $\mathbf{X}^{S_g}$ ($0 \leq L_g \leq L-\sum_{h=1}^{g-1} L_h$). For the marginal distribution of $\mathbf{X}^{S_1}$ the same assumption \textit{A1)} used for the previous models still holds. The remaining assumptions are modified as follows.

\textit{A2*)} For $g=2, \ldots, G$, the conditional p.d.f.~of $\mathbf{X}^{S_g}$ given $(\mathbf{X}^{S_1}, \ldots, \mathbf{X}^{S_{g-1}})$ is equal to
\begin{equation}\label{mnga}
f\left(\mathbf{x}^{S_g}|(\mathbf{x}^{S_1}, \ldots, \mathbf{x}^{S_{g-1}}); \boldsymbol{\theta}_g\right)=\sum_{k=1}^{K_g} \pi^{(g)}_k \phi_{L_g}\left(\mathbf{x}^{S_g};\boldsymbol{\mu}^{(g)}_k,\boldsymbol{\Sigma}^{(g)}_k\right),
\end{equation}
where $K_g \geq 1$,
\begin{equation}\label{mngb}
\boldsymbol{\mu}^{(g)}_k=\boldsymbol{\gamma}^{(g)}_{k} + \sum_{h=1}^{g-1} \mathbf{B}_{gh}\mathbf{x}^{S_h},
\end{equation}
$\mathbf{B}_{gh}$ is a $L_g \times L_h$ matrix, and $\boldsymbol{\theta}_g=(\boldsymbol{\pi}_g, \boldsymbol{\gamma}_g, \mathbf{B}_{g1}, \ldots, \mathbf{B}_{g,g-1}, \boldsymbol{\sigma}_g)$,
$\boldsymbol{\pi}_g=\left(\pi^{(g)}_1,\ldots, \pi^{(g)}_{K_g}\right)$, $\boldsymbol{\gamma}_g=\left(\boldsymbol{\gamma}^{(g)}_1, \ldots, \boldsymbol{\gamma}^{(g)}_{K_g}\right)$,
$\boldsymbol{\sigma}_g=\left(\boldsymbol{\Sigma}^{(g)}_1, \ldots, \boldsymbol{\Sigma}^{(g)}_{K_g}\right)$.
Thus, the dependence of $\mathbf{X}^{S_g}$ on $(\mathbf{X}^{S_1}, \ldots, \mathbf{X}^{S_{g-1}})$
is described by a multivariate linear regression model whose error terms follow a mixture of $K_g$
Gaussian components, with the constraint that the effect of $\mathbf{X}^{S_h}$, $h=1, \ldots, g-1$, on $\mathbf{X}^{S_g}$ is the same for all the $K_g$ components of the mixture (\ref{mnga}).

\textit{A3*)} The conditional distribution of $\mathbf{X}^{U}$ given $\left(\mathbf{X}^{S_1}, \ldots \mathbf{X}^{S_G}\right)$ follows a Gaussian linear regression model; namely:
\begin{equation} \label{regressionG}
\mathbf{X}^{U} = \boldsymbol{\alpha}_{0} + \sum_{g=1}^G \mathbf{A}_{g}\mathbf{x}^{S_g} + \boldsymbol{\epsilon}_U, \ \boldsymbol{\epsilon}_U \sim N_{L_U}\left(\boldsymbol{0},\boldsymbol{\Sigma}_{U}\right).
\end{equation}
Thus, the joint p.d.f.~of $\mathbf{X}$ can be obtained as follows:
\begin{eqnarray} \nonumber
f(\mathbf{x};\boldsymbol{\theta}) & = & \sum_{k=1}^{K_1} \pi^{(1)}_k \phi_{L_1}\left( \mathbf{x}^{S_1};\boldsymbol{\mu}^{(1)}_k,\boldsymbol{\Sigma}^{(1)}_k\right) \\ \nonumber
& \times & \prod_{g=2}^G \left[ \sum_{k=1}^{K_g} \pi^{(g)}_k \phi_{L_g}\left(\mathbf{x}^{S_g};\boldsymbol{\gamma}^{(g)}_{k} + \sum_{h=1}^{g-1} \mathbf{B}_{gh}\mathbf{x}^{S_h},\boldsymbol{\Sigma}^{(g)}_k\right) \right] \\ \label{modello nuovoG}
& \times &  \phi_{L_U}\left(\mathbf{x}^{U}; \boldsymbol{\alpha}_{0} + \sum_{g=1}^G \mathbf{A}_{g}\mathbf{x}^{S_g},\boldsymbol{\Sigma}_U\right),
\end{eqnarray}
where $\boldsymbol{\theta}= \left(\boldsymbol{\theta}_1, \ldots, \boldsymbol{\theta}_G, \boldsymbol{\theta}_U\right)$, with $\boldsymbol{\theta}_U=\left(\boldsymbol{\alpha}_0, \mathbf{A}_1, \ldots, \mathbf{A}_G,
\boldsymbol{\Sigma}_U\right)$.

Similarly to the general models defined by removing restrictions from model (\ref{modello nuovo}), more general models can be defined also from model (\ref{modello nuovoG}). The formal specification of such models is omitted for the ease of presentation.

\subsection{Model identifiability}\label{sec:identifiability}

Some results concerning parameter identifiability are developed for a class of models obtained from equation (\ref{modello nuovoG}). Namely, models characterised by different values of $G$, different partitions of $\mathbf{X}$ into $G$ (or $G+1$) sub-vectors and different values of $K_1, \ldots, K_G$ are examined. Such models can be described from the partitions $(S_1, \ldots, S_G, U)$ of the variable index set $\mathcal{I}=\{1, \ldots, L\}$ (with $U$ possibly equal to $\emptyset$). In particular,
the model class is obtained by admitting that $G \in \mathcal{I}$, $K_g \in \{2, \ldots, K_{g max}\}$
for $g=1, \ldots, G$, and $(S_1, \ldots, S_G, U) \in \mathcal{V}_G$, where $K_{g max}$ denotes the maximum number of components specified by the researcher for the mixture model defined in equation (\ref{mnga}), and $\mathcal{V}_G$ is the family of the
partitions of the variable index set $\mathcal{I}$ into $G$ or $G+1$ elements (as $U$ can be equal to the empty set).
That is, $\mathcal{V}_G=\{ (S_1, \ldots, S_G, U); (S_1, \ldots, S_G, U) \in \mathcal{F}^{G+1}, S_g \ne
\emptyset \ \forall g, S_g \cap S_h = \emptyset \ \forall g \ne h, U= \mathcal{I}\setminus \cup_{g=1}^G S_g \}$,
where $\mathcal{F}$ denotes the family of subsets of $\mathcal{I}$. The resulting model class is denoted as $\mathcal{M}$, and the generic element of $\mathcal{M}$ is denoted as
 $M=(G, S_1, \ldots, S_G, U, K_1, \ldots, K_G)$.
For each model
$M \in \mathcal{M}$ the parameterized densities are denoted by $f(\cdot|\boldsymbol{\theta}_M)$, with
$\boldsymbol{\theta}_M=(\boldsymbol{\theta}_1, \ldots, \boldsymbol{\theta}_G, \boldsymbol{\theta}_U)$.
The corresponding parameter space is denoted by $\mathcal{Q}_{M}$.

Let $M=(G, S_1, \ldots, S_g, S_{g+1}, \ldots, S_G, U, K_1, \ldots, K_g, K_{g+1}, \ldots, K_G)
\in \mathcal{M}$ be a given model under consideration. Then,

\begin{align} \label{modello nuovo,sg,sg+1} \nonumber
f\left(\mathbf{x}^{S_g \cup S_{g+1}}|\mathbf{x}^{\cup_{h=1}^{g-1} S_h}\right)  =
\sum_{k=1}^{K_g} \pi^{(g)}_k \phi_{L_g}\left(\mathbf{x}^{S_g};
 \boldsymbol{\gamma}^{(g)}_{k}+ \sum_{h=1}^{g-1} \mathbf{B}_{gh}\mathbf{x}^{S_h}, \boldsymbol{\Sigma}^{(g)}_k\right)\\
\times
\sum_{k'=1}^{K_{g+1}} \pi^{(g+1)}_{k'} \phi_{L_{g+1}}\left(\mathbf{x}^{S_{g+1}};
 \boldsymbol{\gamma}^{(g+1)}_{k'}+ \sum_{h=1}^{g} \mathbf{B}_{(g+1)h}\mathbf{x}^{S_h}, \boldsymbol{\Sigma}^{(g+1)}_{k'}\right).
\end{align}

Then, consider the Cartesian product $\{1, \ldots, K_g \} \times \{1, \ldots, K_{g+1} \}$ and let the $t$-th element of such a set be denoted as $c_t=(k,k')$, where $t=1, \ldots, K_g \cdot K_{g+1}$. Using this notation it is possible to rewrite equation (\ref{modello nuovo,sg,sg+1}) as follows:
\begin{equation} \label{modello nuovo,sg,sg+1, bis}
f\left(\mathbf{x}^{S_t}|\mathbf{x}^{\cup_{h=1}^{g-1} S_h}\right)=
\sum_{t=1}^{K_t} \pi_t \phi_{L_t}\left(\mathbf{x}^{S_t};
\boldsymbol{\gamma}_t+ \sum_{h=1}^{g-1} \mathbf{B}_{th}\mathbf{x}^{S_h}, \boldsymbol{\Sigma}_{t}\right),
\end{equation}
where $S_t=S_g \cup S_{g+1}$, $K_t=K_g \cdot K_{g+1}$, $L_t=L_g+L_{g+1}$,
\begin{eqnarray}
\pi_t & = &  \pi^{(g)}_k \pi^{(g+1)}_{k'} ,\\
\boldsymbol{\gamma}_{t} & = & \left(\begin{array}{c}                                             \boldsymbol{\gamma}^{(g)}_{k} \\
\boldsymbol{\gamma}^{(g+1)}_{k'} + \mathbf{B}_{(g+1)g}\boldsymbol{\gamma}^{(g)}_{k}
                                            \end{array} \right), \\
\boldsymbol{\Sigma}_t & = & \left[\begin{array}{cc}
                                              \boldsymbol{\Sigma}^{(g)}_k & \boldsymbol{\Sigma}^{(g)}_k\mathbf{B}_{(g+1)g}'\\
                                               \mathbf{B}_{(g+1)g}\boldsymbol{\Sigma}^{(g)}_k
                                               &  \boldsymbol{\Sigma}^{(g+1)}_{k'} + \mathbf{B}_{(g+1)g}\boldsymbol{\Sigma}^{(g)}_k\mathbf{B}_{(g+1)g}'
                                            \end{array} \right],
\end{eqnarray}
and $\mathbf{B}_{th} = \left[\begin{array}{c}
                                            \mathbf{B}_{gh}\\
                                               \mathbf{B}_{(g+1)h} + \mathbf{B}_{(g+1)g}\mathbf{B}_{gh}
                                               \end{array} \right]$.


In this case, for each $\boldsymbol{\theta}_M \in \mathcal{Q}_{M}$ there is a
$\boldsymbol{\theta}^*_{M^*} \in \mathcal{Q}_{M^*}$ such that
$f(\cdot|\boldsymbol{\theta}_M)=f(\cdot|\boldsymbol{\theta}^*_{M^*})$, where $M^*={(G-1, S_1, \ldots, S_t, \ldots, S_G, U, K_1, \ldots, K_t, \ldots, K_G)} \in \mathcal{M}$; thus, the identifiability cannot be ensured. This remark makes it possible to introduce a set of necessary conditions for identifiability. This set is defined in the following theorem.

\begin{thm}
Let $\mathcal{M}'$ be the sub-class of $\mathcal{M}$ composed of the identifiable models. Then,
$\forall$ $M=(G, S_1, \ldots, S_g, \ldots, S_G, U, K_1, \ldots, K_g, \ldots, K_G) \in \mathcal{M}'$,
$\forall g$ $(1 \leq g \leq G)$, it is not possible to find any partition $(s_a,s_b)$
of $S_g$ such that $\forall t$ $(1 \leq t \leq K_g)$ and $ \forall h < g$ the following equalities hold:
\begin{eqnarray}\label{eq:id1}
\pi^{(g)}_t & = & \pi^{(a)}_k \pi^{(b)}_{k'},\\
\boldsymbol{\gamma}^{(g)}_{t} & = & \left(\begin{array}{c}                                             \boldsymbol{\gamma}^{(a)}_{k}\\
\boldsymbol{\gamma}^{(b)}_{k'} + \mathbf{B}_{ba}\boldsymbol{\gamma}^{(a)}_{k}
                                            \end{array} \right), \\
\boldsymbol{\Sigma}^{(g)}_t & = & \left[\begin{array}{cc}
                                              \boldsymbol{\Sigma}^{(a)}_k & \boldsymbol{\Sigma}^{(a)}_k\mathbf{B}_{ba}'\\
                                               \mathbf{B}_{ba}\boldsymbol{\Sigma}^{(a)}_k
                                               & \boldsymbol{\Sigma}^{(b)}_{k'} + \mathbf{B}_{ba}\boldsymbol{\Sigma}^{(a)}_k\mathbf{B}_{ba}'
                                            \end{array} \right],\\ \label{eq:id4}
\mathbf{B}_{gh} & = & \left[\begin{array}{c}
                                            \mathbf{B}_{ah}\\
                                               \mathbf{B}_{bh} + \mathbf{B}_{ba}\mathbf{B}_{ah}
                                               \end{array} \right],
\end{eqnarray}
with $1 \leq k \leq K_{a}$, $1 \leq k' \leq K_{b}$, and $K_{a} \cdot K_{b}= K_{g}$.
\end{thm}

\begin{proof}
Equalities (\ref{eq:id1})-(\ref{eq:id4}) are sufficient conditions for the model $M$
to be unidentifiable. As remarked above, if these equalities hold for a given $g$,
$f\left(\mathbf{x}^{S_g}|\mathbf{x}^{\cup_{h=1}^{g-1} S_h}\right)$ can be written as follows:
\begin{eqnarray} \label{proof} \nonumber
f\left(\mathbf{x}^{S_g}|\mathbf{x}^{\cup_{h=1}^{g-1} S_h}\right) & = & f\left(\mathbf{x}^{S_a}|\mathbf{x}^{\cup_{h=1}^{g-1} S_h}\right) \times f\left(\mathbf{x}^{S_b}|\mathbf{x}^{\cup_{h=1}^{g-1} S_h}, \mathbf{x}^{S_a}\right)\\ \nonumber
& = & \sum_{k=1}^{K_{a}} \pi^{(a)}_k
\phi_{L_a}\left(\mathbf{x}^{s_a};
\boldsymbol{\gamma}^{(a)}_{k} + \sum_{h=1}^{g-1} \mathbf{B}_{a h}\mathbf{x}^{S_h}, \boldsymbol{\Sigma}^{(a)}_k
\right)\\
& \times & \sum_{k'=1}^{K_{b}} \pi^{(b)}_{k'}
\phi_{L_b}\left(\mathbf{x}^{s_b};
\boldsymbol{\gamma}^{(b)}_{k'}  + \sum_{h=1}^{g-1} \mathbf{B}_{b h}\mathbf{x}^{S_h}
+ \mathbf{B}_{b a}\mathbf{x}^{s_a}, \boldsymbol{\Sigma}^{(b)}_{k'}\right),
\end{eqnarray}
where $L_a+L_b=L_g$. Thus, model $M$ is not identified.
\end{proof}

\subsection{Parameter estimation}\label{sec:estimation}

For a given model $M=(G, S_1, \ldots, S_G, U, K_1, \ldots, K_G) \in \mathcal{M}$, whose parameters are  $\boldsymbol{\theta}_M=(\boldsymbol{\theta}_1, \ldots, \boldsymbol{\theta}_G, \boldsymbol{\theta}_U)
\in \mathcal{Q}_{M}$, the estimation can be performed using the ML method. Given a random sample $\mathbf{x}=(\mathbf{x}_1, \ldots,
\mathbf{x}_i, \ldots, \mathbf{x}_n)$ drawn from model $M$, the log-likelihood is
\begin{eqnarray} \nonumber
l(\boldsymbol{\theta}_M) & = &
  \sum_{i=1}^n \ln \left[ \sum_{k=1}^{K_1} \pi^{(1)}_k \phi_{L_1}\left( \mathbf{x}_i^{S_1};\boldsymbol{\mu}^{(1)}_k,\boldsymbol{\Sigma}^{(1)}_k\right) \right]\\ \nonumber
 & + & \sum_{g=2}^G \sum_{i=1}^n \ln \left[
 \sum_{k=1}^{K_g} \pi^{(g)}_k \phi_{L_g}\left(\mathbf{x}_i^{S_g};\boldsymbol{\gamma}^{(g)}_{k} + \sum_{h=1}^{g-1} \mathbf{B}_{gh}\mathbf{x}_i^{S_h},\boldsymbol{\Sigma}^{(g)}_k\right)
 \right]\\ \label{loglik modello nuovo}
 & + &
\sum_{i=1}^n \ln \left[ \phi_{L_U}\left(\mathbf{x}_i^{U}; \boldsymbol{\alpha}_{0} + \mathbf{A}_{1}\mathbf{x}_i^{S_1} + \mathbf{A}_{2}\mathbf{x}_i^{S_2},\boldsymbol{\Sigma}_U\right)\right].
\end{eqnarray}

Thus, since $l(\boldsymbol{\theta}_M)$ is composed of $G+1$ parts, each of which only depends on a sub-vector of $\boldsymbol{\theta}_M$, the maximization of $l(\boldsymbol{\theta}_M)$ can be obtained by a separate maximization of the $G+1$ parts. More specifically, the ML estimation of $\boldsymbol{\theta}_1$ can be computed using the EM algorithm for a Gaussian mixture model \citep[see, e.g.,][]{dempster1977,mclachlan2000}. As far as the ML estimation of $\boldsymbol{\theta}_g$ ($g=2, \ldots, G$) is concerned, it can be carried out by resorting to the EM algorithms developed by \citet{soffritti2011} and \citet{galimberti2015}. The EM algorithm described in the latter paper is also suitable to deal with models resulting from equation~(\ref{mn2c_rid_seemingly}). Finally, $\hat{\boldsymbol{\theta}}_U$ can be computed using the
ML solution for a multivariate linear regression model with Gaussian error terms \citep[see, e.g.,][]{srivastava2002}.

All models for the random vector $\mathbf{X}$ described in Section~\ref{sec:model} are based on the use of Gaussian mixture models whose components are assumed to have unconstrained covariance matrices. Models with a reduced number of variance-covariance parameters can be obtained by reparameterising the component-covariance matrices of any mixture in model (\ref{modello nuovoG}) according to their spectral decomposition and by imposing constraints on the resulting eigenvalues and/or eigenvectors over the mixture components \citep{banfield1993,celeux1995}. Namely, $\boldsymbol{\Sigma}^{(g)}_k=\lambda^{(g)}_{k}\mathbf{D}^{(g)}_{k}\mathbf{A}^{(g)}_{k}\mathbf{D}'^{(g)}_{k}$,
where $\lambda^{(g)}_{k}=|\boldsymbol{\Sigma}^{(g)}_{k}|^{1/L_g}$, $\mathbf{D}^{(g)}_{k}$ is the matrix of eigenvectors of $\boldsymbol{\Sigma}^{(g)}_{k}$ and $\mathbf{A}^{(g)}_{k}$ is the diagonal matrix containing the eigenvalues of $\boldsymbol{\Sigma}^{(g)}_{k}$ (normalized in such a way that $|\mathbf{A}^{(g)}_{k}|=1$). In this parameterisation, the volume, shape and orientation of the $k$th component in the mixture model (\ref{mnga}) are determined by the parameters $\lambda^{(g)}_{k}$, $\mathbf{A}^{(g)}_{k}$ and $\mathbf{D}^{(g)}_{k}$, respectively. Thus, by constraining one or more of these parameters to be the same for all components, for $g=1, \ldots, G$, parsimonious and interpretable models can be obtained. Thus, the model class $\mathcal{M}$ defined in Section~\ref{sec:identifiability} can be widened so as to include such parsimonious models. Namely, a model of this widened class can be denoted as $M^*=(G, S_1, \ldots, S_G,U,K_1,\ldots, K_G,P_1,\ldots, P_G,P_U)$, where $P_g$ ($g=1, \ldots G$) denotes the parameterisation of the component-covariance matrices of the mixture model (\ref{mnga}) and $P_U$ denotes the form (spherical, diagonal, unconstrained) of the covariance matrix $\boldsymbol{\Sigma}_U$ in model~(\ref{regression}). The parsimonious Gaussian mixture models (and their ML estimators) resulting from imposing (up to) fourteen different constraints on such parameters are illustrated in \cite{celeux1995}. They can be estimated using, for example, the package \verb"mclust". Details on the estimation of parsimonious Gaussian clusterwise linear regression models can be found in \cite{dang2015}. Models from this latter class are estimated and compared in Section~\ref{sec:datireali}.

\subsection{Model selection}\label{sec:model_selection}

The selection of an appropriate model in the model space for a given dataset can be performed through the same methods usually employed to select the number of components or the parameterizations of the component-covariance matrices in model-based cluster analysis \citep[see, e.g.,][]{mclachlan2000}. A widely employed information-based criterion is the $BIC$:
\begin{equation}\label{BIC}
BIC_M = 2 l(\hat{\boldsymbol{\theta}}_M) - npar_M \ln(n),
\end{equation}
where $l(\hat{\boldsymbol{\theta}}_M)$ and $npar_M$ denote the maximum value of the log-likelihood and number of estimated parameters in model $M$, respectively. Note that for models $M$ with a log-likelihood equal to the one defined in equation (\ref{loglik modello nuovo}) $BIC_M$ can be obtained by summing the $BIC$ values associated with the $G+1$ parts of $l(\boldsymbol{\theta}_M)$.

In a Bayesian framework, $BIC$ provides an asymptotic approximation of the log-posterior probability of a model \citep{kass1995}. The use of this criterion can be motivated on the basis
of both theoretical and practical results. \citet{keribin2000} proves that using the $BIC$ allows to consistently estimate
the number of mixture components. The criteria for performing variable selection in Gaussian model-based cluster analysis proposed by \citet{maugis2009a} and \citet{maugis2009b}, based on the $BIC$, are proved to be consistent under regularity conditions. A similar result is proved in \citet{galimberti2013} for selecting the partition
of the variables in a parsimonious approach to model-based cluster analysis.
From an applied point of view, good performances of the $BIC$ as a model selection criterion for
Gaussian mixture models are reported in several papers, such as \citet{biernacki1999} and \citet{fraley2002}. All results illustrated in Section~\ref{sec:results} are obtained by using this criterion.

\section{Experimental results}\label{sec:results}

\subsection{Further results from the introductory example}\label{sec:ripresa esempio intro}

All models illustrated in Section~\ref{sec:intro_example} can be seen as a special case of the models described in Section~\ref{sec:model}. They are listed in Table~\ref{t:intro example}, where $\mathbf{X}=(X_1, X_2, X_3)$ and $X_1$, $X_2$ and $X_3$ represent mother's height, father's height and student's height, respectively. Model $M_7$ is an improvement of model $M_6$ resulting from the elimination of $X_1$ from the regressors in the model for $\mathbf{X}^{S_2}$.

\begin{table*}
\caption{Maximised log-likelihood and $BIC$ value of seven models fitted to the heights dataset.}
\label{t:intro example}
\centering
\begin{tabular}{lcccccccccc}
\hline\noalign{\smallskip}
Models & $\mathbf{X}^{S_1}$ & $\mathbf{X}^{S_2}$ & $\mathbf{X}^{U}$ & $K_1$ & $K_2$ &  $\mathbf{X}_U^{S_1}$ & $\mathbf{X}_{S_2}^{S_1}$ &  $l(\hat{\boldsymbol{\theta}}_M)$ & $npar_M$ & $BIC_M$ \\
\noalign{\smallskip}\hline\noalign{\smallskip}
$M_1$ &  $\mathbf{X}$  & $\emptyset$ & $\emptyset$ &  1 & - &  - & - & $-$235.51    &     9        & $-$502.48\\
$M_2$ &  $X_1, X_2$  & $\emptyset$ & $X_3$ & 2 & - & $X_1,X_2$ & - &  $-$231.82     &     11        & $-$502.10\\
$M_3$ &  $X_1, X_2$  & $\emptyset$ & $X_3$ & 2 & - & $X_2$ & - & $-$232.68     &     10        & $-$500.33\\
$M_4$ &  $X_1, X_2$  & $X_3$ & $\emptyset$ & 1 & 1 & - & $X_2$ & $-$238.00     &     6        & $-$496.98\\
$M_5$ &  $X_1, X_2$  & $X_3$ & $\emptyset$ & 1 & 2 & - & $X_2$ & $-$235.15     &     8        & $-$498.27\\
$M_6$ &  $X_1, X_2$  & $X_3$ & $\emptyset$ & 2 & 2 & - & $X_1, X_2$ & $-$228.46     &     13        & $-$502.38\\
$M_7$ &  $X_1, X_2$  & $X_3$ & $\emptyset$ & 2 & 2 & - & $X_2$ & $-$229.83     &     12        & $-$501.63\\
\noalign{\smallskip}\hline
\end{tabular}
\end{table*}

According to the $BIC$, the best model is $M_4$. It describes the marginal distribution of the parents' heights with a Gaussian model in which the two heights are treated as independent and with the same variance; furthermore, the fathers' height is used as a regressor in a Gaussian linear regression model for describing the conditional distribution of the students' height. Thus, no evidence of cluster structure seems to be found in the dataset, neither in the distribution of the parents' heights nor in the conditional distribution of the students' height.
However, the difference between the $BIC$ of model $M_4$ and the one of other fitted models is low \citep{kass1995}. Thus, models supporting the presence of a cluster structure in the marginal distribution of the parents' heights ($M_3$), in the conditional distribution of the students' height ($M_5$) and in both distributions ($M_7$) may be employed for this dataset.

\subsection{Exploring the model space using genetic algorithms}\label{sec:algoritmo genetico}

In order to find the optimal model in the model space for a given dataset, all possible models have to be considered and compared. However, an approach based on such an exhaustive search is computationally expensive and time-consuming,
especially when the number of observed variables is high. Thus, non-exhaustive strategies are needed. A solution is represented by genetic algorithms. They constitute stochastic optimisation techniques that exploit principles and operators of the biological evolution of a species for solving complex problems with a vast number of possible solutions \citep{goldberg1989}. These algorithms are widely used in many fields of statistics \citep[see, e.g.,][]{chatterjee1996}. Applications in subset selection problems can be found, for example, in \citet{bozdogan2004}.

In general, a genetic algorithm starts from the examination of the chromosomes (ordered sequences of genes) that compose an initial population. Each of these chromosomes is randomly generated; it is assigned a value summarising its fitness. Then, an iterative evolution process is performed, based on three main genetic operators (selection, crossover, mutation), with the goal of generating novel populations composed of chromosomes characterised by improved fitness values. The selection operator consists in a weighted random sampling from the initial population with weights that are generally proportional to the chromosomes' fitness. The chromosomes selected in this way reproduce and their offspring will compose a novel generation. Such a generation is obtained after crossover and mutation. Namely, crossover is a random process of genome recombination that applies to pairs of chromosomes; mutation is a random alteration of a gene in a chromosome. The chromosomes of the resulting novel generation are assigned their fitness and the evolution process repeats. Usually, the algorithm stops when either a maximum number of populations has been generated or a satisfactory value of the fitness has been reached for the population.

A first genetic algorithm is developed for finding the model $\hat{M}$ such that
\begin{equation}\label{eq:optim_function}
\hat{M}=\argmax_{M \in \mathcal{M}} BIC_M,
\end{equation}
where $\mathcal{M}$ is the class of models illustrated in Section~\ref{sec:identifiability} having up to two cluster structures ($G=2$). In this algorithm each model $M$ is represented as a chromosome whose fitness is given by $BIC_M$. The evolution process is composed of two parts:
\begin{description}
  \item[\textit{a})] information extraction for the specification of model (\ref{mn1});
  \item[\textit{b})] information extraction for the specification of models (\ref{mn2a}) and (\ref{regression}).
  \end{description}
In part \textit{a}) the examined chromosomes have a binary gene for each variable in $\mathbf{X}$ (with 1/0 denoting a variable selected/not selected for $\mathbf{X}^{S_1}$, respectively); they also have two genes associated with the number of components to be used in models~(\ref{mn1}) and (\ref{mn2a}). The possible values of these two latter genes are the integer numbers between 1 and a maximum value chosen by the researcher ($K_{1max}$ for model~(\ref{mn1}), $K_{2max}$ for model~(\ref{mn2a})). Thus, chromosomes of length $L+2$ are examined in the first part of the algorithm. Let $M_\textit{a}$ be the best model detected at the end of part \textit{a}) and let $\hat{S}_1, \hat{K}_1$ be the solution for $S_1, K_1$ obtained from $M_\textit{a}$.

If $\hat{L}_1 < L$, where $\hat{L}_1$ denotes the length of $\mathbf{X}^{\hat{S}_1}$, in the second part of the genetic algorithm a solution for $S_2$, $K_2$ and $U$ is searched by keeping fixed the solution for $S_1, K_1$ obtained from the model $M_\textit{a}$. Specifically, the chromosomes are composed of $L-\hat{L}_1+1$ genes, with a binary gene for each variable in $\mathbf{X} \smallsetminus \mathbf{X}^{\hat{S}_1}$ (where 1 and 0 denote a variable selected and not selected for $\mathbf{X}^{S_2}$, respectively), and an additional gene (with positive integer values) associated with the number of components for model~(\ref{mn2a}). The model $M_{\textit{b}}$ obtained at the end of this second part provides $\hat{S}_2,\hat{K}_2,\hat{U}$, where $\hat{U}= \mathcal{I} \smallsetminus \hat{S}_1 \smallsetminus \hat{S}_2$. Thus, $\hat{U}$ will not be empty when $\hat{L}_1 + \hat{L}_2 < L$, with $\hat{L}_2$ denoting the length of $\mathbf{X}^{\hat{S}_2}$.
The final solution of the algorithm is $\hat{M}=(\hat{S}_1,\hat{S}_2,\hat{U},\hat{K}_1,\hat{K}_2)$. All models examined and estimated with this genetic algorithm have unconstrained covariance matrices. The results of an evaluation of its effectiveness, based on simulated datasets generated by models with unconstrained covariance matrices, are provided in Section~\ref{sec:montecarlo}.

A second genetic algorithm is developed for carrying out the model search in the wider class of parsimonious models illustrated in Section~\ref{sec:estimation} with $G=2$. The search is still decomposed in two parts that are similar to those described above; the only difference is in the structure of the chromosomes. Namely, in part \textit{a}) the examined chromosomes have two additional genes (with positive integer values up to 14) for distinguishing the parsimonious parameterisations to be used in models~(\ref{mn1}) and (\ref{mn2a}). In part \textit{b}) an additional gene is necessary to denote the parsimonious parameterisation in model~(\ref{mn2a}); another gene (with three possible values) is added for indicating the form of the covariance matrix $\boldsymbol{\Sigma}_U$ in model~(\ref{regression}).
Section~\ref{sec:datireali} shows the results obtained from the analysis of two real datasets based on this second algorithm.

These genetic algorithms have been implemented in \verb"R" by exploiting the package \verb"GA" \citep{scrucca2013}. Each execution requires the specification of the following tuning parameters: $K_{1max}$, $K_{2max}$, $N_1$ and $N_2$ (dimension of the examined populations in parts \textit{a}) and \textit{b}), respectively), $d_{1max}$ and $d_{2max}$ (maximum number of generations to be examined in the two parts of the algorithm). The specific values of these tuning parameters employed in the analyses are detailed in the following Sections. In both algorithms linear-rank selection and single point crossover operators are used. The probability of crossover between pairs of chromosomes is set equal to $0.8$ in all analyses. As far as the mutation is concerned, this genetic transformation is randomly carried out in a parent chromosome with a probability of $0.1$.

\subsection{A Monte Carlo study}\label{sec:montecarlo}

The performance of the first genetic algorithm is evaluated through a Monte Carlo experiment in which artificial
datasets are generated in the Euclidean space $\mathds{R}^8$ using model~(\ref{modello nuovo}), where $\mathbf{X}^{S_{1}}= (X_1, X_2, X_3)$, $K_1=2$, $\mathbf{X}^{S_{2}}= (X_4, X_5, X_6)$, $K_2=2$,
and $\mathbf{X}^{U}= (X_7, X_8)$.

Specifically, the parameters of the marginal p.d.f.~of $\mathbf{X}^{S_1}$ are: $\pi^{(1)}_1=0.5$,
\begin{equation}\nonumber
\boldsymbol{\mu}^{(1)}_{1}=\left(
                                  \begin{array}{c}
    0\\
    0\\
    0\\
                                    \end{array}
                                   \right), \
\boldsymbol{\mu}^{(1)}_{2}=\left(
                                  \begin{array}{c}
    5\\
    -5\\
    5\\
                                    \end{array}
                                \right), \
\boldsymbol{\Sigma}^{(1)}_1= \left(
                                  \begin{array}{ccc}
                                    1 & -0.6 & -0.3 \\
                                    -0.6 & 1 & -0.4 \\
                                    -0.3 & -0.4 & 1 \\
                                    \end{array}
                                \right), \
\end{equation}
\begin{equation}\nonumber
\boldsymbol{\Sigma}^{(1)}_2= \left(
                                  \begin{array}{ccc}
                                      1 & 0.6 & 0.3 \\
                                    0.6 & 1 & 0.4 \\
                                    0.3 & 0.4 & 1 \\
                                    \end{array}
                                \right).
\end{equation}

The parameters of the conditional p.d.f.~of $\mathbf{X}^{S_2}$ given $\mathbf{X}^{S_1}$ are: $\pi^{(2)}_1=0.5$,
\begin{equation}\nonumber
\boldsymbol{\gamma}^{(2)}_{1}=\left(
                                  \begin{array}{c}
-2\\
-1\\
3.5\\
                                    \end{array}
                                \right),  \
\boldsymbol{\gamma}^{(2)}_{2}=\left(
                                  \begin{array}{c}
4\\
5\\
-2.5\\
\end{array}
                                \right), \
\mathbf{B}_{21} = \left(
                    \begin{array}{ccc}
                                    1.5 & 2 & 1.5 \\
                                    1.5 & -2.5 & -2 \\
                                    1.5 & 2 & -2.5 \\
                                    \end{array}
                                \right), \
\end{equation}
\begin{equation}\nonumber
\boldsymbol{\Sigma}^{(2)}_1= \left(
                                  \begin{array}{ccc}
                                    1 & 0.5 & 0.6 \\
                                    0.5 & 1 & 0.4 \\
                                    0.6 & 0.4 & 1 \\
                                    \end{array}
                                \right),
\boldsymbol{\Sigma}^{(2)}_2= \left(
                                  \begin{array}{ccc}
                                    1 & -0.5 & -0.6 \\
                                    -0.5 & 1 & -0.4 \\
                                    -0.6 & -0.4 & 1 \\
                                    \end{array}
                                \right).
\end{equation}
Finally, the parameters of the conditional p.d.f.~of $\mathbf{X}^{U}$ given
$(\mathbf{X}^{S_1}, \mathbf{X}^{S_2})$ are: $\boldsymbol{\alpha}_0=\left(
                                                                 \begin{array}{c}
                                                                 2 \\
                                                                 2 \\
                                                             \end{array}
                                \right)$, $\mathbf{A}_1= \left(
                                  \begin{array}{ccc}
                                    2 & 2 & 2 \\
                                    -2 & -2 & -2 \\
                                    \end{array}
                                \right)$,
$\mathbf{A}_2=-\mathbf{A}_1$ and $\boldsymbol{\Sigma}_U= \left(
                                  \begin{array}{cc}
                                    2.25 & 0 \\
                                    0 & 1 \\
                                    \end{array}
                                \right)$.

The main goal of this experiment is to evaluate the effectiveness of the genetic algorithm (with the $BIC$ as a fitness measure) in selecting the model the datasets come from. To this end, one hundred samples of $n=200$ observations each are generated and analysed using the algorithm. Since the algorithm's performance may depend on how large is the exploration of the model space, the algorithm is executed by changing the values of the tuning parameters $N_1$ and $d_{1max}$ that control the information extraction for the specification of model (\ref{mn1}). Namely, the examined values are 120, 160 and 200 for $N_1$; 30, 40 and 50 for $d_{1max}$. The other tuning parameters are kept constant throughout the experiment; they are set as follows: $K_{1max}=K_{2max}=3$, $N_2=80$ and $d_{2max}=20$. Greater values of $N_2$ and $d_{2max}$ are not examined because some preliminary analyses have highlighted that increasing them has a little impact on the results. This is mainly due to the fact that in part \textit{b}) of the algorithm the exploration of the model space is carried out conditionally on the results obtained in part \textit{a}). Thus, nine executions of the algorithm are performed. The same analysis is carried out with the sample size $n=400$.

The effectiveness of the genetic algorithm is evaluated with respect to the ability to recover the correct variable partition. The percentage of datasets for which $\mathbf{X}^{S_1}$, $\mathbf{X}^{S_2}$ and $\mathbf{X}^{U}$ are successfully identified is generally high, especially with $n=400$ (see Tables~\ref{tab:sim1} and \ref{tab:sim2}, first row). The most common error is represented by a partition in which both uninformative variables are allocated to the variable sub-vector that defines the first cluster structure, and $X_2$ is wrongly inserted in the variable sub-vector that defines the second cluster structure. Whenever the exploration of the model space is widened, an improvement in the effectiveness of the algorithm is generally obtained. This improvement is almost exclusively associated with an increase in the tuning parameter $N_1$. Using a value greater than 30 for the maximum number of examined generations, in association with any examined value of $N_1$, is completely useless with $n=400$.

\begin{table}[t]
\caption{Distribution of 100 simulated datasets according to the partition of $\mathbf{X}$ obtained with the genetic algorithm in the executions with $n=200$.} \label{tab:sim1}
\centering
\begin{tabular}{lccccccccc} \hline
& \multicolumn{9}{c}{$N_1$}\\
 & \multicolumn{3}{c}{120} & \multicolumn{3}{c}{160} & \multicolumn{3}{c}{200}\\\cline{2-10}
 & \multicolumn{9}{c}{$d_{1max}$} \\
 & 30 & 40 & 50 & 30 & 40 & 50 & 30 & 40 & 50 \\ \hline
Correct classification of all variables       &  76 & 79 & 80 & 91 &  88 & 91  & 91  & 91 &  90\\
Correct recovery of $\mathbf{X}^{S_{1}}$ only & 4   & 3 &  3 & 3 &  4  & 3 & 3 & 3 &  4\\
$\hat{S}_{1}=(1, 3, 7, 8)$, $\hat{S}_{2}=(2, 4, 5, 6)$ & 14   & 14 & 13 & 5 &  8  & 5 & 6 & 6 &  5\\
Other wrong partitions & 6   & 4 & 4 & 1 &  0  & 1 & 0 & 0 &  1\\
\hline
\end{tabular}
\end{table}

\begin{table}[t]
\caption{Distribution of 100 simulated datasets according to the partition of $\mathbf{X}$ obtained with the genetic algorithm in the executions with $n=400$.} \label{tab:sim2}
\centering
\begin{tabular}{lccccccccc} \hline
& \multicolumn{9}{c}{$N_1$}\\
 & \multicolumn{3}{c}{120} & \multicolumn{3}{c}{160} & \multicolumn{3}{c}{200}\\\cline{2-10}
 & \multicolumn{9}{c}{$d_{1max}$} \\
 & 30 & 40 & 50 & 30 & 40 & 50 & 30 & 40 & 50 \\ \hline
Correct classification of all variables       &  87 & 87 & 87 & 92 & 92 & 92 & 95  & 95 & 95 \\
Correct recovery of $\mathbf{X}^{S_{1}}$ only & 1   & 1 &  1 & 2 & 2   & 2 & 2 & 2 & 2 \\
$\hat{S}_{1}=(1, 3, 7, 8)$, $\hat{S}_{2}=(2, 4, 5, 6)$ & 9   & 9 & 9 & 5 & 5  & 5& 3 & 3 & 3\\
Other wrong partitions & 3   & 3 & 3 & 1 & 1  & 1& 0 & 0 & 0 \\
\hline
\end{tabular}
\end{table}

The effectiveness is also evaluated with respect to the ability of the genetic algorithm to recover the two latent cluster structures. This task is carried out by computing the adjusted Rand index between the true cluster structures and the structures estimated by the algorithm. Although in some datasets the true variable partition is not correctly detected, the first cluster structure identified by the genetic algorithm perfectly coincides with the first true one in all samples (aRi=1). Also the agreement between the second true cluster structure and the second estimated one is very high: the mean (over 100 datasets) of the aRi is greater that 0.996 in all executions of the algorithm for both sample sizes. These results are due to the fact that there is a clear-cut separation between clusters in both $\mathbf{X}^{S_1}$ and $\mathbf{X}^{S_2}$; furthermore, in all datasets in which the genetic algorithm selects wrong variable partitions, two variables in both $\mathbf{X}^{\hat{S}_1}$ and $\mathbf{X}^{\hat{S}_2}$ are always correctly selected.

\subsection{Results from real datasets}\label{sec:datireali}

Examples are carried out using two benchmark real datasets: the crabs dataset and the AIS dataset. The crabs dataset is described in \citet{venables2002} and is available in the \verb"R" package \verb"MASS".
It reports $L=5$ morphological measurements (in mm) for $n=200$ crabs of the species Leptograpsus variegatus: frontal lobe size (FL), rear width (RW), carapace length (CL), carapace width (CW) and body depth (BD). Namely, the sample is composed of 50 crabs each of two colours (blue and orange) and both sexes. The AIS dataset is described in Cook and Weisberg (1994) and is available in the \verb"R" package \verb"sn" \citep{azzalini2014}. It contains information concerning $n=202$ athletes (102 males and 100 females) at the Australian Institute of Sport. In particular, the analysis described in this Section focuses on $L=9$ variables: red cell count (RCC), white cell count (WCC), hematocrit (Hc), hemoglobin (Hg), plasma ferritin concentration (Fe), body mass index (BMI), sum of skin folds (SSF), body fat percentage (Bfat) and lean body mass (LBM). Each dataset is analysed using the second genetic algorithm illustrated in Section~\ref{sec:algoritmo genetico}. For comparison purposes, analyses are carried out also using the \verb"R" packages \verb"mclust" and \verb"clustvarsel" and the softwares \verb"SelvarClust" and \verb"SelvarClustIndep".

\subsubsection{Crabs dataset}\label{sec:crabs}

The best Gaussian mixture model fitted to the joint p.d.f.~of the five morphological measurements using \verb"mclust" (with a maximum number of components equal to five) is a mixture of four Gaussian ellipsoidal components with the same volume and shape. This result is obtained with an option for the initialisation of the EM algorithm that transforms the variables using a singular value decomposition. The value of $BIC$ for such a mixture model is $-2842.3$. Table~\ref{tab:crabs1} reports the cross classification of the crabs based on the maximum estimated posterior probabilities and their gender and colour. The clustering obtained from this model reproduces quite well the four classes of crabs defined from gender and colour.

\begin{table}
\caption{Classification of the crabs according to their colour and gender (BF = blue female, BM = blue male, OF = orange female, OM = orange male) and the cluster membership estimated by the model selected using \texttt{mclust}.}  \label{tab:crabs1}
\centering
\begin{tabular}{lcccc}
  \hline
 & \multicolumn{4}{l}{Colour and gender} \\
Cluster & BF & BM & OF & OF \\ \hline
1 & 49 & 11 & 0  & 0\\
2 & 0  & 0  & 5  & 50\\
3 & 0  & 39 & 0  & 0 \\
4 & 1  & 0  & 45 & 0\\ \hline
aRi & \multicolumn{4}{c}{0.794} \\  \hline
\end{tabular}
\end{table}

According to the variable selection methods implemented in \verb"clustvarsel" only four morphological measurements are relevant for clustering the crabs. They are frontal lobe size, rear width, carapace width and body depth. The best Gaussian mixture model fitted to the p.d.f.~of these measurements is a mixture of four Gaussian ellipsoidal components with the same volume and shape. Using this model allows to obtain a partition of the crabs with an increased agreement with the partition based on gender and colour (see Table~\ref{tab:crabs2}). The same result is also obtained using both \verb"SelvarClust" and \verb"SelvarClustIndep".
The $BIC$ value of the resulting joint model for the five measurements, obtained as described in Section~\ref{sec:estimation}, is equal to $-2811.2$, thus leading to an improved model with respect to the one obtained without variable selection.

\begin{table}
\caption{Classification of the crabs according to their colour and gender and the cluster membership estimated by the model selected using \texttt{clustvarsel}.}  \label{tab:crabs2}
\centering
\begin{tabular}{lcccc}
  \hline
 & \multicolumn{4}{l}{Colour and gender}\\
Cluster   & BF & BM & OF & OF \\ \hline
1 & 50 & 10 & 0  & 0\\
2 & 0  & 0  & 5  & 50\\
3 & 0  & 40 & 0  & 0 \\
4 & 0  & 0  & 45 & 0\\ \hline
aRi & \multicolumn{4}{c}{0.815}\\   \hline
\end{tabular}
\end{table}

The splitting of the five measurements produced by two independent executions of the genetic algorithm (using the tuning parameters $N_1=N_2=400$, $d_{1max}=d_{2max}=40$, $K_{1max}=K_{2max}=5$ in the first execution and $N_1=N_2=500$, $d_{1max}=d_{2max}=50$, $K_{1max}=K_{2max}=5$ in the second) is: $\mathbf{X}^{\hat{S}_1}=$ (RW, CL), $\mathbf{X}^{\hat{S}_2}=$ (FL, CW, BD), $\mathbf{X}^{\hat{U}}=\emptyset$. A mixture of two Gaussian components is selected for both the joint marginal p.d.f.~of (RW, CL) and the conditional p.d.f.~of (FL, CW, BD)$|$(RW, CL).
Namely, the two components are ellipsoidal with the same volume and shape in the mixture model for (RW, CL), while they are spherical with the same volume in the model for (FL, CW, BD)$|$(RW, CL). The $BIC$ value of the resulting joint model is $-2812.7$. Thus, two cluster structures are detected. Tables~\ref{tab:crabs3} and~\ref{tab:crabs4} show some classifications of the crabs pertaining to the first and the second cluster structure, respectively. The clustering of crabs resulting from the analysis of rear width and carapace length reproduces quite well the classification of crabs based on their gender, while it is not associated with the classification based on colour. On the contrary, an almost perfect agreement arises from the comparison between this latter classification and the clustering obtained by modelling
the dependence of frontal lobe size, carapace width and body depth on rear width and carapace length using a clusterwise Gaussian linear regression model with two components. The second cluster structure detected in the dataset does not take into account differences in gender.

\begin{table}
\caption{Comparison between the first cluster structure detected by the genetic algorithm in the joint marginal distribution of RW and CL and the classifications of crabs based on their colour and/or gender.}  \label{tab:crabs3}
\centering
\begin{tabular}{lcccccccccc}
  \hline
& \multicolumn{4}{l}{Colour and gender} & & \multicolumn{2}{l}{Colour} & & \multicolumn{2}{l}{Gender}\\
Cluster         & BF & BM & OF & OF  & & B  & O  & & F   & M \\ \hline
      1 & 50 & 7 & 50  & 3   & & 57 & 53 & & 100 & 10\\
      2 & 0  & 43  & 0  & 47  & & 47 & 43 & & 0 & 90\\
  \hline
      aRi & \multicolumn{4}{c}{0.400} & & \multicolumn{2}{c}{$-$0.003} & & \multicolumn{2}{c}{0.810} \\
  \hline
\end{tabular}
\end{table}

\begin{table}
\caption{Comparison between the second cluster structure detected by the genetic algorithm in the conditional distribution of (FL, CW, BD)$|$(RW, CL) and the classifications of crabs based on their colour and/or gender.}  \label{tab:crabs4}
\centering
\begin{tabular}{lcccccccccc}
  \hline
 & \multicolumn{4}{l}{Colour and gender} & & \multicolumn{2}{l}{Colour} & & \multicolumn{2}{l}{Gender}\\
Cluster        & BF & BM & OF & OF  & & B  & O  & & F   & M \\ \hline
      1 & 50 & 50 & 1  & 0   & & 100 & 1 & & 51 & 50\\
      2 & 0  & 0  & 49  & 50  & & 0 & 99 & & 49 & 50\\
  \hline
    aRi & \multicolumn{4}{c}{0.486} & & \multicolumn{2}{c}{0.980} & & \multicolumn{2}{c}{$-$0.005} \\
  \hline
\end{tabular}
\end{table}

The model selected by the genetic algorithm assumes that FL, CW and BD linearly depend on both RW and CL. Since this assumption could be restrictive, models are also estimated in which a different set of regressors is allowed for each of the three regression equations in the multivariate linear regression model of frontal lobe size, carapace width and body depth. According to the $BIC$, the best solution obtained after examining these further models is a mixture of two seemingly unrelated linear regression models in which frontal lobe size and carapace width are both regressed on carapace length and rear width, while body depth only depends on carapace length. The clustering of the crabs resulting from such a model coincides with the classification of crabs based on their colour (aRi=1). The joint model for the five morphological measurements given by the product of this seemingly unrelated linear regression model and the Gaussian mixture model for CL and RW described above has a $BIC$ value ($-2808.3$) which is slightly higher than the best model obtained performing model-based variable selection. This model allows to detect the presence of two different cluster structures in the dataset, each of which is strongly or perfectly associated with one of the two known classifications that characterise this dataset. As a consequence of above, the two detected cluster structures are independent one another.

\subsubsection{AIS dataset}\label{sec:ais}

The best Gaussian mixture model resulting from the analysis performed through \verb"mclust" (with a maximum number of components equal to five and using the same initialisation of the EM algorithm employed in the analysis of the crabs dataset) is a mixture of three Gaussian ellipsoidal components with the same orientation. The value of $BIC$ for such a mixture model is $-9028.2$. The clustering obtained from this model reproduces quite well the two classes of athletes based on their gender (see Table~\ref{tab:ais1}, left part).

The three examined variable selection methods lead to different decisions about the variables that provide relevant information on the clustering of the athletes (see Table~\ref{tab:ais2}). Using \verb"clustvarsel" (with a maximum number of components for the p.d.f.~of the informative variables equal to five), only the biometrical variables are selected. The best Gaussian mixture model fitted to the p.d.f.~of these variables is a mixture of three Gaussian ellipsoidal components with the same shape. The $BIC$ value of the resulting joint model for the nine variables is $-9008.1$. Thus, according to the $BIC$, this joint model is better than the best model detected without variable selection. However, the partition of the athletes resulting from the best mixture model for the biometrical variables shows a slightly lower agreement with the partition based on gender (see Table~\ref{tab:ais1}).

\begin{table}
\caption{Classification of the athletes according to their gender and the cluster membership estimated by the models selected using \texttt{mclust}, \texttt{clustvarsel}, \texttt{SelvarClust} and \texttt{SelvarClustIndep}.}  \label{tab:ais1}
\centering
\begin{tabular}{lcccccccccccccccc}
  \hline
&  \multicolumn{3}{l}{\texttt{mclust}} & & \multicolumn{3}{l}{\texttt{clustvarsel}} & & \multicolumn{3}{l}{\texttt{SelvarClust}} & & \multicolumn{4}{l}{\texttt{SelvarClustIndep}}\\ \hline
  &  \multicolumn{3}{l}{Cluster}      & & \multicolumn{3}{l}{Cluster}      & & \multicolumn{3}{l}{Cluster}      & & \multicolumn{3}{l}{Cluster} &\\
 Gender & 1 & 2 & 3                                & & 1 & 2 & 3                                & & 1 & 2 & 3                        & & 1 & 2 & 3 &\\ \hline
F & 97 & 2 & 1 & & 39 & 61 & 0                             & & 2 & 1 & 97                       & & 2 & 97 & 1 &\\
M & 2 & 40 & 60 & & 13 & 1 & 88                             & & 25 & 75 & 2                      & & 23 & 2 & 77 &\\  \hline
aRi & \multicolumn{3}{c}{0.682}             & & \multicolumn{3}{c}{0.586}             & & \multicolumn{3}{c}{0.735} & & \multicolumn{3}{c}{0.745} & \\  \hline
\end{tabular}
\end{table}

\begin{table}
\caption{Variables selected by the packages \texttt{clustvarsel}, \texttt{SelvarClust} and \texttt{SelvarClustIndep} from the AIS dataset.}  \label{tab:ais2}
\centering
\begin{tabular}{ll}
\hline
Package & Selected variables \\ \hline
\texttt{clustvarsel} & BMI, SSF, Bfat, LBM \\
\texttt{SelvarClust} & BMI, SSF, Bfat, LBM, Fe, Hg \\
\texttt{SelvarClustIndep} & BMI, SSF, Bfat, LBM, Fe, Hc \\ \hline
\end{tabular}
\end{table}

Softwares \texttt{SelvarClust} and \texttt{SelvarClustIndep} also select two blood composition variables: one is plasma ferritin concentration and the other is hemoglobin or hematocrit.
Namely, the best joint model obtained after three independent executions of \texttt{SelvarClust} is given by the product of a Gaussian mixture model with three equally-oriented components for the joint marginal distribution of BMI, SSF, Bfat, LBM, Fe and Hg, and a Gaussian linear regression model for the conditional distribution of the remaining variables in which only Hg is used as a regressor and the covariance matrix is unconstrained. The $BIC$ value of the joint model obtained in this way is $-8935.3$. Using \texttt{SelvarClustIndep} the best model is composed of a mixture
of $K=3$ Gaussian components with the same orientation for the joint marginal distribution of BMI, SSF, Bfat, LBM, Fe and Hc, and a Gaussian linear regression model for the conditional distribution of the remaining variables in which the selected regressors are Hc, BMI and Bfat and the covariance matrix is diagonal. None of the uninformative
variables results to be independent of all the informative ones.
Overall, this joint model registers a $BIC$ of $-8934.5$. As far as the recovery of the classification based on gender is concerned, the partitions of the athletes resulting from the mixture models for the variables selected by \texttt{SelvarClust} and \texttt{SelvarClustIndep} reach a very similar performance, that is slightly better than the ones obtained using both \texttt{mclust} and \texttt{clustvarsel} (see Table~\ref{tab:ais1}).

\begin{table}
\caption{Cluster structures detected by the genetic algorithm and their association with the classification of the athletes based on gender.}  \label{tab:ais3}
\centering
\begin{tabular}{lcccccccc}
  \hline
  & \multicolumn{4}{l}{Structure 1}  & \multicolumn{4}{l}{Structure 2}\\
 & \multicolumn{3}{l}{Cluster} & & \multicolumn{3}{l}{Cluster}  & \\
Gender       & 1 & 2 & 3                   & & 1 & 2 & 3 &\\ \hline
    F & 98 & 1 & 1                   & & 50 & 49 & 1 &\\
    M & 1  & 74 & 27                 & & 54 & 34 & 14 &\\
\hline
aRi & \multicolumn{3}{c}{0.754}     & & \multicolumn{3}{c}{0.015}\\
  \hline
\end{tabular}
\end{table}

The results obtained from two independent executions of the genetic algorithm (using the tuning parameters $N_1=N_2=300$, $d_{1max}=d_{2max}=30$, $K_{1max}=K_{2max}=4$ in the first execution and $N_1=N_2=700$, $d_{1max}=d_{2max}=70$, $K_{1max}=K_{2max}=5$ in the second) highlight the existence of two cluster structures. Namely, the first structure is found in the sub-vector $\mathbf{X}^{\hat{S}_1}=$ (BMI, SSF, Bfat, LBM, Hg). The best model for the p.d.f.~of this sub-vector is a mixture of three Gaussian ellipsoidal components with the same orientation.
The recovery of males and females classes obtained using the segmentation of the athletes based on this model is slightly improved (see Table~\ref{tab:ais3}, left part). The second cluster structure is found in the conditional distribution of the sub-vector $\mathbf{X}^{\hat{S}_2}=$ (Hc, Fe) given $\mathbf{X}^{\hat{S}_1}$, resulting from a mixture of three Gaussian components with diagonal covariance matrices having the same volume. The partition of the athletes obtained from this second mixture model is not associated with the athletes' gender (see Table~\ref{tab:ais3}, right part). Since in model (\ref{modello nuovo}) the latent variables $Z_1$ and $Z_2$ are assumed to be independent, this latter result is not surprising. Thus, the second structure is reasonably associated with other (unobserved) factors independent of the gender. Finally, red and white cell counts compose $\mathbf{X}^{\hat{U}}$ and, thus, result to be uninformative variables. Their conditional distribution is modelled using a Gaussian linear regression model with a diagonal regression covariance matrix. The $BIC$ value of the joint model for the nine variables is $-8933.2$.

An improvement of the clusterwise linear regression model with three components selected for (Hc, Fe) is obtained after performing regressors selection. This task is carried out by allowing each dependent variable to have its own specific set of regressors. Furthermore, all fourteen parsimonious parameterisations are estimated for each examined model. According to the $BIC$, the best solution is obtained using a model in which haematocrit is regressed on hemoglobin, sum of skin folds and body fat percentage, while the selected predictors for plasma ferritin concentration are hemoglobin, body mass index, body fat percentage and lean body mass. The component-covariance matrices of this model are unconstrained. In a similar way, the best Gaussian linear regression model for (RCC, WCC) is the one that has haematocrit as a predictor for both dependent variables and sum of skin folds only for WCC. The covariance matrix of this model is diagonal. Since the joint model for the nine variables obtained in this way has a $BIC$ value of $-8856.9$, it seems to provide a good description of the relevant information contained in the AIS dataset.

\section{Conclusions}\label{sec:conclusions}

The proposed framework relies on a very general model that makes it possible to compare the results of different (supervised and unsupervised) analyses carried out on a given dataset. Namely, this model allows to perform variable selection in model-based clustering according to the methods proposed by \citet{raftery2006}, \citet{maugis2009a} and \citet{maugis2009b}. It also allows to carry out model selection in multivariate linear regression analysis and seemingly unrelated linear regression analysis by assuming either a Gaussian model or a Gaussian mixture model for the distribution of the errors \citep{soffritti2011, galimberti2015}. Furthermore, using the proposed model enables the detection of the presence of multiple cluster structures from possibly correlated variable sub-vectors.

It should be noted that the process of selecting a model in this framework may be quite complex and may prevent the methodology illustrated in this paper from being used with high-dimensional datasets. Resorting to genetic algorithms can partly mitigate this drawback. This kind of algorithms is able to globally explore the model space, thus avoiding the main problems that typically characterise stepwise or greedy search strategies. Clearly, the effectiveness of a genetic algorithm greatly depends on how large is the performed exploration of the model space. However, no general rule about how to choose proper values of the population size and the number of generations in a genetic algorithm is known. The Monte Carlo study summarised in Section~\ref{sec:montecarlo} suggests that an important role is played by the population size. A proper choice of the tuning parameters in a genetic algorithm as well as other aspects concerning model selection (e.g.: a comparison with other strategies for exploring the model space; how to deal with high-dimensional datasets) represent topics for future research. Nevertheless, the experimental results illustrated in Sections~\ref{sec:ripresa esempio intro} and \ref{sec:datireali} show that for some datasets the joint use of supervised and unsupervised learning methods allows to extract unknown relevant information that otherwise would be missed, thus supporting the usefulness of the proposed framework.

\bibliographystyle{elsarticle-harv}

\begin{thebibliography}{}

\bibitem[Andrews and McNicholas(2014)]{andrews2014}
Andews, J.L., McNicholas, P.D.: Variable selection for clustering and classification. J. Classif. {\bf 31}, 136--153 (2014)

\bibitem[Azzalini(2014)]{azzalini2014}
Azzalini, A.: The \verb"R" package \verb"sn": The skew-normal and skew-$t$
distributions (version 1.1-2). \texttt{URL http://azzalini.stat.unipd.it/SN} (2014)


\bibitem[Banfield and Raftery(1993)]{banfield1993}
Banfield, J.D., Raftery, A.E.: Model-based Gaussian and non-Gaussian clustering.
Biometrics, {\bf 49}, 803-–821 (1993)

\bibitem[Belitskaya-Levy(2006)]{belitskaya2006}
Belitskaya-Levy, I.: A generalized clustering problem, with application to DNA microarrays. Stat. Appl. Genet. Mol. Biol. {\bf 5}, Article 2 (2006)

\bibitem[Biernacki and Govaert(1999)]{biernacki1999}
Biernacki, C., Govaert, G.: Choosing models in model-based clustering and discriminant analysis.
J. Stat. Comput. Simul. {\bf 64}, 49--71 (1999)

\bibitem[Bozdogan(2004)]{bozdogan2004}
Bozdogan, H.: Intelligent statistical data mining with information complexity and genetic algorithms. In:
Bozdogan, H. (ed.), Statistical Data Mining and Knowledge Discovery. Chapman \& Hall/CRC, London, 15--56 (2004)

\bibitem[Brusco and Cradit(2001)]{brusco2001}
Brusco, M.J., Cradit, J.D.: A variable-selection heuristic for k-means clustering.
Psychometrika {\bf 66}, 249--270 (2001)

\bibitem[Celeux and Govaert(1995)]{celeux1995}
Celeux, G., Govaert, G.: Gaussian parsimonious clustering models.
Pattern Recognit. {\bf 28}, 781-–793 (1995)

\bibitem[Celeux \textit{et al.}(2011)]{celeux2011}
Celeux, G., Martin-Magniette, M-L., Maugis, C., Raftery, A.E.:
Letter to the editor. J. Am. Stat. Assoc. {\bf 106}, 383 (2011)

\bibitem[Celeux \textit{et al.}(2014)]{celeux2014}
Celeux, G., Martin-Magniette, M-L., Maugis-Rabusseau, C., Raftery, A.E.:
Comparing model selection and regularization
approaches to variable selection in model-based
clustering. J. Soc. Fr. Statistique {\bf 155}, 57--71 (2014)

\bibitem[Chatterjee \textit{et al.}(1996)]{chatterjee1996}
Chatterjee, S., Laudato, M., Lynch, L.A.: Genetic algorithms and their statistical applications: an introduction.
Comput. Stat. Data Anal. {\bf 22}, 633--651 (1996)

\bibitem[Cook and Weisberg(1994)]{cook1994}
Cook, R.D., Weisberg, S.: An Introduction to Regression Graphics.
Wiley, New York (1994)

\bibitem[Dang and Bailey(2015)]{dangbailey2015}
Dang, X.H., Bailey, J.: A framework to uncover multiple alternative clusterings.
Mach. Learn. {\bf 98}, 7--30 (2015)

\bibitem[Dang and McNicholas(2015)]{dang2015}
Dang, U.J., McNicholas, P.D.: Families of parsimonious finite mixtures of regression models. In:
Morlini, I., Minerva T., Vichi, M. (eds.), Statistical Models for Data Analysis. Springer, Berlin, 73-84 (2015)

\bibitem[Dempster \textit{et al.}(1977)]{dempster1977}
Dempster, A.P., Laird, N.M., Rubin, D.B.: Maximum likelihood for
incomplete data via the EM algorithm. J. R. Stat. Soc. Ser. B {\bf 39},
1–-22 (1977)

\bibitem[De Sarbo and Cron(1988)]{desarbo1988}
De Sarbo, W.S., Cron, W.L.:
A maximum likelihood methodology for clusterwise linear regression.
J. Classif. {\bf 5}, 249--282 (1988)

\bibitem[De Veaux(1989)]{deveaux1989}
De Veaux, R.D.: Mixtures of linear regressions. Comput. Stat. Data Anal. {\bf 8}, 227--245 (1989)

\bibitem[Dy and Brodley(2004)]{dy2004}
Dy, J.G., Brodley, C.E.: Feature selection for unsupervised learning.
J. Mach. Learn. Res. {\bf 5}, 845--889 (2004)

\bibitem[Fowlkes \textit{et al.}(1988)]{fowlkes1988}
Fowlkes, E.B., Gnanadesikan, R., Kettenring, J.R.: Variable selection in clustering.
J. Classif. {\bf 5}, 205--228 (1988)

\bibitem[Fraiman \textit{et al.}(2008)]{fraiman2008}
Fraiman, R., Justel, A., Svarc, M.: Selection of variables for cluster analysis and classification rules.
J. Am. Stat. Assoc. {\bf 103}, 1294--1303 (2008)

\bibitem[Fraley and Raftery(2002)]{fraley2002}
Fraley, C., Raftery, A.E.: Model-based clustering,
discriminant analysis and density estimation. J. Am. Stat. Assoc. {\bf 97}, 611--631 (2002)

\bibitem[Fraley \textit{et al.}(2012)]{fraley2012}
Fraley, C., Raftery, A.E., Murphy, T.B., Scrucca, L.: \verb"mclust" version 4 for \verb"R": normal mixture modeling for model-based clustering, classification, and density estimation. Technical Report No. 597, Department of Statistics, University of Washington (2012)

\bibitem[Friedman and Meulman(2004)]{friedman2004}
Friedman, J.H., Meulman, J.J.: Clustering objects on subsets of attributes (with discussion).
J. R. Stat. Soc. Ser. B {\bf 66}, 815--849 (2004)

\bibitem[Galimberti \textit{et al.}(2009)]{galimberti2009}
Galimberti, G., Montanari, A., Viroli, C.: Penalized factor mixture analysis for variable selection in clustered data.
Comput. Stat. Data Anal. {\bf 53}, 4301--4310 (2009)

\bibitem[Galimberti and Soffritti(2007)]{galimberti2007}
Galimberti, G., Soffritti, G.: Model-based methods to identify multiple cluster structures in a data set.
Comput. Stat. Data Anal. {\bf 52}, 520--536 (2007)

\bibitem[Galimberti and Soffritti(2013)]{galimberti2013}
Galimberti, G., Soffritti, G.: Using conditional independence for parsimonious model-based Gaussian clustering.
Stat. Comput. {\bf 23}, 625--638 (2013)

\bibitem[Galimberti \textit{et al.}(2015)]{galimberti2015}
Galimberti, G., Scardovi, E., Soffritti, G.: Using mixtures in seemingly unrelated linear regression models with non-normal errors.
Stat. Comput. doi=10.1007/s11222-015-9587-0, 1-14 (2015)

\bibitem[Gnanadesikan \textit{et al.}(1995)]{gnanadesikan1995}
Gnanadesikan, R., Kettenring, J.R., Tsao, S.L.: Weighting and selection of variables
for cluster analysis. J. Classif. {\bf 12}, 113--136 (1995)

\bibitem[Goldberg(1989)]{goldberg1989}
Goldberg, D.E.: Genetic Algorithms in Search, Optimization, and Machine Learning. Addison-Wesley, Reading (1989)

\bibitem[Gordon(1999)]{gordon1999}
Gordon, A.D.: Classification, 2nd edn. Chapman \& Hall, Boca Raton (1999)

\bibitem[Guo \textit{et al.}(2010)]{guo2010}
Guo, J., Levina, E., Michailidis, G., Zhu, J.: Pairwise variable selection for high-dimensional
model-based clustering. Biometrics {\bf 66}, 793--804 (2010)

\bibitem[Hastie \textit{et al.}(2009)]{hastie2009}
Hastie, T., Tibshirani, R., Friedman, J.: The Elements of Statistical Learning: Data Mining, Inference, and Prediction, 2nd edn. Springer, New York (2009)

\bibitem[Hubert and Arabie(1985)]{hubert1985}
Hubert, L., Arabie, P.: Comparing partitions. J. Classif. {\bf 2} 193--218 (1985)

\bibitem[Kass and Raftery(1995)]{kass1995}
Kass, R.E., Raftery, A.E.: Bayes factors.
J. Am. Stat. Assoc. {\bf 90}, 773--795 (1995)

\bibitem[Keribin(2000)]{keribin2000}
Keribin, C.: Consistent estimation of the order of mixture models.
Sankhy\={a} Ser. A {\bf 62}, 49--66 (2000)

\bibitem[Law \textit{et al.}(2004)]{law2004}
Law, M.H.C., Figueiredo, M.A.T., Jain, A.K. Simultaneous feature selection and
clustering using mixture models. IEEE Trans. Pattern Anal. Mach. Intell. {\bf 26}, 1154--1166 (2004)

\bibitem[Liu \textit{et al.}(2015)]{liu2015}
Liu, T., Zhang, N.L., Chen, P. Liu, A.H., Poon, L.K.M, Wang, Y.: Greedy learning of latent tree models for multidimensional clustering. Mach. Learn. {\bf 98}, 301--330 (2015)

\bibitem[Maugis \textit{et al.}(2009a)]{maugis2009a}
Maugis, C., Celeux, G., Martin-Magniette, M-L.: Variable selection for clustering
with Gaussian mixture models. Biometrics {\bf 65}, 701--709 (2009a)

\bibitem[Maugis \textit{et al.}(2009b)]{maugis2009b}
Maugis, C., Celeux, G., Martin-Magniette, M-L.: Variable selection in model-based clustering:
a general variable role modeling. Comput. Stat. Data Anal. {\bf 53}, 3872--3882 (2009b)

\bibitem[McLachlan and Peel(2000)]{mclachlan2000}
McLachlan, G.J., Peel, D.: Finite Mixture Models. John Wiley \& Sons, Chichester (2000)

\bibitem[Melnykov and Maitra(2010)]{melnykov2010}
Melnykov, V., Maitra, R.: Finite mixture models and model-based clustering.
Stat. Surv. {\bf 4}, 80--116 (2010)

\bibitem[Montanari and Lizzani(2001)]{montanari2001}
Montanari, A., Lizzani, L.: A projection pursuit approach to variable selection.
Comput. Stat. Data Anal. {\bf 35}, 463--473 (2001)

\bibitem[Pan and Shen(2007)]{pan2007}
Pan, W., Shen, X.: Penalized model-based clustering with application to variable selection.
J. Mach. Learn. Res. {\bf 8}, 1145--1164 (2007)

\bibitem[Poon \textit{et al.}(2013)]{poon2013}
Poon, L.K.M., Zhang, N.L., Liu, T., Liu, A.H.: Model-based clustering of high-dimensional data:
variable selection versus facet determination. Int. J. Approx. Reason. {\bf 54}, 196--215 (2013)


\bibitem[Quandt and Ramsey(1978)]{quandt1978}
Quandt, R.E., Ramsey, J.B.: Estimating mixtures of normal distributions and switching
regressions. J. Am. Stat. Assoc. {\bf 73}, 730-738 (1978)


\bibitem[Raftery and Dean(2006)]{raftery2006}
Raftery, A.E., Dean, N.: Variable selection for model-based cluster analysis.
J. Am. Stat. Assoc. {\bf 101}, 168--178 (2006)

\bibitem[R Core Team(2015)]{R2015}
R Core Team: \verb"R": a language and environment for statistical computing.
\verb"R" Foundation for Statistical Computing, Vienna, Austria.
\texttt{URL http://www.R-project.org} (2015)

\bibitem[Schwarz(1978)]{schwarz1978}
Schwarz, G.: Estimating the dimension of a model.
Ann. Stat. {\bf 6}, 461--464 (1978)

\bibitem[Scrucca(2013)]{scrucca2013}
Scrucca, L.: \verb"GA": a package for genetic algorithms in \verb"R".
J. Stat. Softw. {\bf 53}, 4 (2013)

\bibitem[Soffritti(2003)]{soffritti2003}
Soffritti, G.: Identifying multiple cluster structures in a data matrix. Commun. Stat. Simul.
{\bf 32}, 1151--1177 (2003)

\bibitem[Soffritti and Galimberti(2011)]{soffritti2011}
Soffritti, G., Galimberti, G.:
Multivariate linear regression with non-normal errors: a solution based on mixture models.
Stat. Comput. {\bf 21}, 523--536 (2011)

\bibitem[Srivastava(2002)]{srivastava2002}
Srivastava, M.S.: Methods of Multivariate Statistics. John Wiley \& Sons, New York (2002)

\bibitem[Steinley and Brusco(2008a)]{steinley2008a}
Steinley, D., Brusco, M.J.: a new variable weighting and selection procedure for k-means cluster analysis.
Multivar. Behav. Res. {\bf 43}, 77--108 (2008a)

\bibitem[Steinley and Brusco(2008b)]{steinley2008b}
Steinley, D., Brusco, M.J.: Selection of variables in cluster analysis: An empirical comparison of eight
procedures. Psychometrika {\bf 73}, 125--144 (2008b)

\bibitem[Tadesse \textit{et al.}(2005)]{tadesse2005}
Tadesse, M.G., Sha, N., Vannucci, M.: Bayesian variable selection in clustering high-dimensional data.
J. Am. Stat. Assoc. {\bf 100}, 602--617 (2005)

\bibitem[Venables and Ripley(2002)]{venables2002}
Venables, W.N., Ripley, B.D.: Modern Applied Statistics with S. Fourth edition. Springer, New York (2002)

\bibitem[Wang and Zhu(2008)]{wang2008}
Wang, S., Zhu, J.: Variable selection for model-based high-dimensional clustering and
its application to microarray data. Biometrics {\bf 64}, 440--448 (2008)

\bibitem[Witten and Tibshirani(2010)]{witten2010}
Witten, D.M., Tibshirani, R.: A framework for feature selection in clustering.
J. Am. Stat. Assoc. {\bf 105}, 713--726 (2010)

\bibitem[Xie \textit{et al.}(2008)]{xie2008}
Xie, B., Pan, W., Shen, X.: Variable selection in penalized model-based clustering via
regularization on grouped parameters. Biometrics {\bf 64}, 921--930 (2008)

\bibitem[Zeng and Cheung(2009)]{zeng2009}
Zeng, H., Cheung, Y.-M.: A new feature selection method for Gaussian mixture clustering.
Pattern Recognit. {\bf 42}, 243--250 (2009)

\bibitem[Zhou \textit{et al.}(2009)]{zhou2009}
Zhou, H., Pan, W., Shen, X.: Penalized model-based clustering with unconstrained covariance
matrices. Electron. J. Stat. {\bf 3}, 1473--1496 (2009)
\end{thebibliography}








\end{document}